\documentclass{elsarticle}

\usepackage[utf8]{inputenc}
\usepackage{geometry}
\usepackage{amsmath, amssymb, amsthm, amsfonts}
\usepackage{graphicx}
\usepackage{hyperref}
\usepackage{tikz}
\usepackage{multirow}
\usetikzlibrary{fit, positioning}
\usepackage{booktabs}
\usepackage{algorithm}
\usepackage{algpseudocode}
\usepackage{tabularx}

\newtheorem{theorem}{Theorem}[section]
\newtheorem{lemma}[theorem]{Lemma}
\newtheorem{assumption}{Assumption}[section]

\begin{document}

\begin{frontmatter}

\title{Differentiable Normative Guidance for Nash Bargaining Solution Recovery}


\author[1]{Moirangthem Tiken Singh\corref{cor1}}
\ead{tiken.m@dibru.ac.in} 

\author[1]{Surajit Borkotokey}
\ead{sborkotokey@dibru.ac.in}

\author[2]{Rajnish Kumar}
\ead{rajnish.kumar@qub.ac.uk}

\cortext[cor1]{Corresponding author}

\affiliation[1]{organization={Dibrugarh University}, 
             city={Dibrugarh}, 
             state={Assam}, 
             country={India}}

\affiliation[2]{organization={Queen's University}, 
             city={Belfast}, 
             country={United Kingdom}}

\begin{abstract}

Autonomous artificial intelligence agents in negotiation systems must generate equitable utility allocations satisfying individual rationality (IR), ensuring each agent receives at least its outside option, and the Nash Bargaining Solution (NBS), which maximizes joint surplus. Existing generative models often learn suboptimal human behaviors, producing solutions far from Pareto efficiency, while classical methods require full Pareto frontier knowledge, which is unavailable in real datasets. We propose a guided graph diffusion framework that generates individually rational utility vectors while approximating the NBS without frontier knowledge at inference time. Negotiations are modeled as directed graphs with graph attention capturing asymmetric agent attributes, and a conditional diffusion model maps these to utility vectors. A differentiable composite guidance loss, applied in the final reverse diffusion steps, penalizes IR violations and Nash product gaps. We prove that, under sufficient penalty weighting, solutions enter the IR region in finite time. Across datasets, the method achieves 100\% IR compliance. Nash efficiency reaches 99.45\% on synthetic data (within 0.55 percentage points of an oracle), and 54.24\% (CaSiNo) and 88.67\% (Deal or No Deal), improving 20–60 percentage points over unconstrained generative baselines.
\end{abstract}


\end{frontmatter}

\section{Introduction}
\label{sec:introduction}

Autonomous artificial intelligence (AI) agents agents are increasingly being deployed as intermediaries in high-stakes resource allocation processes, including supply chain contracting, spectrum allocation in wireless networks, legal dispute resolution, and international diplomatic coordination \cite{jonker2017automated, vaccaro2025advancing}. In each of these domains, the agent is tasked with proposing a division of utility that multiple stakeholders may either accept or reject. The primary challenge in such applications is not solely computational tractability, given the maturity of solvers for constrained optimization problems, but rather the achievement of normative trustworthiness. Specifically, the system must generate outcomes that are simultaneously individually rational (no participant receives less than their outside option), Pareto efficient (no feasible reallocation can make at least one participant better off without making another worse off), and consistent with well justified axiomatic fairness criteria. Absent these formal guarantees, such systems cannot be considered suitable for responsible deployment in high-stakes environments, irrespective of their average performance on held out evaluation data.

To satisfy these stringent requirements, cooperative game theory offers a precise normative benchmark. The seminal contribution of \cite{nash1950bargaining} established that four axioms, namely Pareto efficiency, symmetry, scale invariance, and independence of irrelevant alternatives, uniquely characterize NBS.

This solution maximizes the product of the agents' surplus utilities above their respective disagreement points. It implicitly incorporates individual rationality as a necessary condition and selects the particular Pareto efficient allocation that optimally trades off joint gains. Subsequent work \cite{rubinstein1982perfect} showed that the behavior of rational, strategically interacting agents in an alternating offers bargaining game converges to this solution in the limit as their discount factors approach one, thereby grounding the concept in both axiomatic and noncooperative foundations. Thus, satisfaction of individual rationality together with consistency with NBS forms the core set of requirements for a trustworthy autonomous negotiating agent.

Despite the conceptual precision of this theoretical objective, the dominant methodological paradigm in contemporary artificial intelligence (AI) negotiation research is structurally misaligned with these game-theoretic guarantees. Prevailing approaches predominantly rely on training models directly on corpora of human negotiation dialogues. However, human negotiators exhibit bounded rationality~\cite{simon1955behavioral}: they systematically anchor on initial offers \cite{tversky1974judgment}, display loss-averse preferences under uncertainty \cite{kahneman1979prospect}, and frequently terminate interactions by accepting early agreements rather than continuing to explore potentially more favorable negotiated outcomes.

These behavioral regularities induce empirical utility allocations that are heavily concentrated at interior point outcomes, substantially below the Pareto frontier. When supervised or generative models are trained to approximate these empirical distributions via divergence minimization, they internalize and reproduce the suboptimality, power asymmetries, and cognitive biases embedded in the training data \cite{ntoutsi2020bias, dinnar2021ai}. Prior work further demonstrates that systems explicitly engineered with fairness oriented objectives can nonetheless reinstantiate demographic disparities when assessed in ecologically valid settings \cite{kramar2022negotiation}. 

The empirical implications are pronounced. On standard human negotiation corpora such as CaSiNo \cite{chawla-etal-2021-casino}, widely used generative baselines fail almost entirely to optimize for joint surplus, attaining near zero efficiency. This failure arises not from insufficient representational capacity of the underlying architectures, but from their faithful reproduction of a training distribution that is systematically displaced from the theoretical optimum.

Two fundamental structural limitations underlie this failure in current generative models. First, they do not satisfy game theoretic axioms by design. Post hoc projection of an already biased generative sample onto an axiomatic solution set is inadequate. When the initial sample is concentrated in the vicinity of the disagreement point, a purely radial projection to the Pareto frontier cannot recover the directional information necessary to identify the surplus maximizing allocation. Achieving axiomatic compliance instead requires an iterative guidance mechanism integrated into the generative process itself, such that the sampling trajectory is actively steered toward axiom consistent regions of the outcome space. Second, dominant architectures represent negotiation states as flat vectors. This encoding obscures inter agent relational asymmetries, violates permutation invariance, and impedes rigorous counterfactual analysis \cite{renting2024general}.

To address these specific limitations, this work introduces a guided graph diffusion framework for equitable utility generation. The architecture integrates three principal components designed to mitigate empirical bias. First, a strategic graph encoder represents the negotiation as a directed graph, which is processed by a graph attention network \cite{brinkmann2023machine} to produce a relational context embedding.

This embedding captures asymmetric agent attributes (disagreement points, budgets, priority weights) while preserving permutation invariance at the graph level. Second, a conditional diffusion process employs this embedding to condition a denoising network that generates joint utility vectors \cite{song2021denoising}. Third, a normative guidance mechanism introduces a differentiable composite loss function that simultaneously penalizes Nash product shortfalls, violations of individual rationality, and deviations from the Pareto frontier. To enhance training stability, this gradient based correction is applied only during the final portion of the reverse diffusion process. This late activation schedule enables the generative model to first learn a coarse structure that is proximate to the Pareto set before enforcing fine grained, directionally informed optimization.

The proposed framework is assessed on three datasets that jointly span a spectrum from analytical tractability to empirical complexity: a synthetic corpus for which optimal solutions are exactly computable, the CaSiNo human negotiation corpus \cite{chawla-etal-2021-casino}, and the Deal or No Deal corpus \cite{lewis2017dealornodeal}. Across all evaluated domains, the guided model is mathematically designed to strongly enforce individual rationality constraints, with theoretical backing under defined conditions. In addition, it attains near oracle Nash efficiency on the synthetic benchmark and substantially outperforms the strongest unconstrained generative baselines on human negotiation datasets, yielding large absolute gains in performance without requiring access to the exact Pareto frontier geometry at inference time.

This work offers the following key contributions:

\begin{enumerate}

\item Guided Graph Diffusion for Utility Generation: A generative architecture coupling a relational graph encoder with a conditional diffusion process and a differentiable guidance loss. The system produces utility allocations that satisfy individual rationality by construction and align closely with optimal bargaining solutions across both synthetic and empirical domains.

\item Differentiable Normative Guidance Loss: The formulation of a three term composite guidance loss utilizing smooth penalties. This ensures everywhere nonzero gradients for constraint enforcement and surplus maximization within the reverse diffusion process, supported by formal convergence analysis establishing finite time entry into the feasible set.

\item Activation Window Guidance Mechanism: An empirically and theoretically validated design strategy that restricts normative guidance to the final fraction of reverse diffusion steps. Multidimensional sensitivity analysis confirms this window prevents early stage gradient over constraint and serves as the primary determinant of high generation efficiency.

\item Systematic Empirical Evaluation: A comprehensive benchmark across three datasets, six baseline methods, and four evaluation metrics. The results demonstrate that the proposed framework uniquely satisfies game theoretic axioms and escapes the distributional bias of human negotiation corpora without relying on post hoc mathematical solvers.

\end{enumerate}

Section \ref{sec:background_and_problem} establishes the game theoretic foundations and surveys prior work in automated negotiation, generative modeling, and algorithmic fairness. Section \ref{sec:methodology} formalizes the proposed framework, including the graph encoding, diffusion process, guidance mechanism, and theoretical guarantees. Section \ref{sec:experimental_setup} details the datasets, architectural configurations, and evaluation metrics. Section \ref{sec:results} presents the sensitivity analysis, axiomatic compliance results, latent trajectory dynamics, and baseline comparisons. Section \ref{sec:discussion} interprets the empirical findings and discusses systemic limitations. Finally, Section \ref{sec:conclusion} summarizes the contributions and outlines directions for future research.

\section{Background and Problem Statement}
\label{sec:background_and_problem}

This section establishes the theoretical foundations on which the proposed framework rests and surveys the prior literature to formalize the computational problem. Cooperative game theory provides the normative reference point against which all prior computational approaches are measured. The formal goal of AI mediated negotiation is the construction of an equitable utility allocation across multiple agents in a non transferable utility environment. In these settings, agents cannot freely exchange utility and the feasible space of agreements is determined by the underlying resource structure. Any outcome a rational agent would accept must satisfy individual rationality, requiring that no agent receives less than their respective disagreement point or outside option \cite{fisher1981getting}. An outcome failing this condition is voluntarily rejected. Individual rationality alone admits highly asymmetric allocations.

NBS provides the unique Pareto efficient allocation that maximizes the Nash product \cite{nash1950bargaining, rubinstein1982perfect}. The work in \cite{myerson1984two} extended these foundational results to settings with incomplete information, establishing conditions under which this specific mathematical solution remains the appropriate fairness target even when agents possess private valuations.

Translating these game theoretic guarantees into computational systems has produced distinct strategies over the past decades. The taxonomy provided in \cite{baarslag2013evaluating} distinguishes concession based approaches from utility maximizing approaches. The Automated Negotiating Agents Competition surveyed in \cite{jonker2017automated} has yielded competitive agents for controlled multi-round settings. However, these classical agents share a structural dependency by assuming direct access to the opponent utility function or estimating it online through repeated interaction. This dependency is incompatible with the single-shot generation setting considered in this work, where a utility allocation must be produced directly from agent feature representations. Furthermore, classical solvers require explicit analytical knowledge of the Pareto frontier geometry to identify the optimum, which is unavailable for real world corpora like CaSiNo and Deal or No Deal.

To overcome the frontier knowledge requirement inherent to classical solvers, models are increasingly trained directly on human dialogue corpora. The sequence to sequence architecture in \cite{lewis2017dealornodeal} demonstrated that end-to-end agents can produce valid agreements. However, training on human data replicates suboptimal empirical behaviors rather than optimizing joint utility. The work in \cite{he2018decoupled} identified this distributional absorption problem and proposed decoupling strategy from language to reduce stylistic overfitting. This separation does not resolve utility bias because the strategy module still learns from empirically flawed human negotiated outcomes. Reinforcement learning via reward signal fine tuning \cite{ziegler2019fine} can shift output distributions, but agents often adopt deceptive strategies when rewards do not strictly penalize individual rationality violations \cite{cao2018emergent}. Furthermore, supervised training on historically inequitable data amplifies pre existing disparities \cite{ntoutsi2020bias}. AI agents trained on historical bargaining records reproduce documented demographic disparities in final outcomes \cite{dinnar2021ai}, and systems with explicit fairness objectives still reinstantiate these disparities in ecologically valid settings \cite{kramar2022negotiation}.

As alternatives to direct regression, generative models learn a distribution over outputs. Conditional Variational Autoencoders \cite{sohn2015learning} encode the target output as a latent variable conditioned on an input context. The evidence lower bound objective provides only a lower bound on the true log likelihood, and the approximation gap widens in high variance settings. Conditional Generative Adversarial Networks \cite{mirza2014conditional} learn conditional distributions through adversarial training but are subject to mode collapse and training instability. Neither architecture provides mechanisms to enforce axiomatic constraints on generated outputs, as their respective reconstruction and adversarial objectives are mathematically agnostic to game theoretic compliance.

Addressing the instability of prior generative architectures, Denoising Diffusion Probabilistic Models \cite{ho2020denoising} reformulate generative modeling as a learned reverse process that gradually denoises a noise initialized sample via a score matching objective. Denoising Diffusion Implicit Models \cite{song2021denoising} provide a deterministic non Markovian variant that achieves comparable sample quality with fewer inference steps. Steering diffusion samples toward desired properties during inference was introduced via classifier guidance \cite{dhariwal2021diffusion} and later classifier free guidance \cite{ho2022classifier}. These methods were developed exclusively for continuous perceptual attributes in image generation and do not provide mechanisms for enforcing combinatorial or game theoretic constraints. Diffusion models have been applied successfully to molecular structure \cite{hoogeboom2022equivariant}, but they have not been adapted for game theoretic utility generation utilizing differentiable axiomatic guidance.

The necessity of axiomatic guidance is further underscored by the limitations of standard algorithmic fairness frameworks. The machine learning fairness literature proposes distinct operationalizations of equitable treatment. Individual fairness requires similar individuals to receive similar outcomes \cite{dwork2012fairness}. Group level criteria focus on error rate parity across demographic groups \cite{hardt2016equality}. Counterfactual fairness requires outcome invariance under interventions on protected attributes \cite{kusner2017counterfactual}. These statistical and causal frameworks differ entirely from the axiomatic criterion grounded in cooperative game theory. An outcome can satisfy demographic parity while simultaneously violating individual rationality for a specific agent. Trustworthy AI in high-stakes deployment contexts requires normative constraints that are formally specified and computationally enforced at inference time \cite{vaccaro2025advancing}. 

Consequently, the precise formulation of equitable utility vectors in AI mediated negotiation systems requires satisfying these strict game theoretic axioms in non transferable utility settings. The computational task is to generate an outcome that guarantees strict individual rationality, ensuring each agent achieves a payoff greater than or equal to their outside option, while simultaneously maximizing the product of the surplus gains achieved by all negotiating parties. The primary methodological challenge is executing this constrained generation without analytical access to the Pareto frontier geometry and without internalizing the suboptimal distributional biases present in historical training corpora. The system must produce these normatively compliant utility vectors in a single-shot inference process using only the asymmetric contextual features of the negotiating agents. Resolving this overarching problem requires a novel generative mechanism that decouples state representation from probability modeling and enforces mathematical axioms dynamically during generation.

\section{Methodology}
\label{sec:methodology}

This section introduces a formal framework for the construction of equitable utility vectors, $\mathbf{u} \in \mathbb{R}^n$ with $n \ge 2$, in the context of AI-mediated negotiations. To attenuate biases embedded in historically supervised datasets, we recast the utility-generation problem as a guided stochastic process that is rigorously constrained by foundational axioms from game theory.

\begin{figure}[htbp]
\centering
\begin{tikzpicture}[
    >=stealth,
    node distance=1.0cm and 1 cm,
    font=\sffamily\small,
    block/.style={
        rectangle, draw, rounded corners,
        fill=blue!5, text width=2.6cm,
        align=center, minimum height=1.2cm},
    emb/.style={
        rectangle, draw,
        fill=green!10, text width=1.6cm,
        align=center, minimum height=0.8cm},
    mlp/.style={
        rectangle, draw, rounded corners,
        fill=purple!7, text width=3.8cm,
        align=center, minimum height=1.4cm},
    temb/.style={
        rectangle, draw, rounded corners,
        fill=orange!8, text width=2.4cm,
        align=center, minimum height=0.9cm},
    loss/.style={
        rectangle, draw, rounded corners,
        fill=red!10, text width=3.2cm,
        align=center, minimum height=1.2cm,
        dashed},
    arrow/.style={->, thick},
    guidance/.style={->, thick, dashed, red!70!black}
]

\node[block, text width=2.6cm] (features)
    {Agent Features $\mathbf{x}_1, \mathbf{x}_2 \in \mathbb{R}^k$\\
     {\scriptsize negotiation attributes}};

\node[block, right=of features] (gat)
    {GATv2 Layer\\
     {\scriptsize $n$-node graph}\\
     {\scriptsize $K$ heads}};

\node[emb, right=0.9cm of gat] (h)
    {$h \in \mathbb{R}^{d_h}$};

\node[left=1.4cm of gat, yshift=-2.8cm] (noise)
    {Noise $\mathbf{u}_T \sim \mathcal{N}(\mathbf{0},\mathbf{I})$};

\node[temb, below=2.4cm of gat] (temb)
    {Sinusoidal\\Time Emb.\\
     {\scriptsize $t_{\text{emb}} \in \mathbb{R}^{d_t}$}};

\node[draw, circle, fill=gray!10, inner sep=2pt,
      right=0.9cm of h, yshift=-2.8cm] (cat)
    {\scriptsize cat};

\node[mlp, right=1.4cm of cat] (mlp)
    {MLP Denoiser\\
     $s_\theta(\mathbf{u}_t,\, t,\, h)$\\
     {\scriptsize Linear$(d_z \to d_{\text{mlp}} \to n)$}};

\node[below=1.6cm of mlp] (u0)
    {Utility $\mathbf{u}_0$};

\node[loss, above=1.6cm of mlp] (norm)
    {Normative Guidance\\
     $\mathcal{L}_{\text{guide}}$};
\node[above=0.15cm of norm, font=\scriptsize\sffamily]
    {Nash Barrier $+$ Strict IR};

\draw[arrow] (features) -- (gat);
\draw[arrow] (gat)      -- (h);
\draw[arrow] (h.east)   -- ++(0.3,0) |- (cat.north);
\draw[arrow] (noise.east) -- ++(0.3,0) |- (cat.west);
\draw[arrow] (temb.east) -- ++(0.3,0) |- (cat.south);

\node[font=\scriptsize, above=0.05cm of cat, xshift=-0.9cm]
    {$[\mathbf{u}_t \| t_{\text{emb}} \| h]$};
\node[font=\scriptsize, below=0.05cm of cat, xshift=0.3cm]
    {\scriptsize $\mathbb{R}^{d_z}$};

\draw[arrow] (cat) -- (mlp);
\draw[arrow] (mlp) -- node[right] {Sample} (u0);
\draw[guidance] (norm) --
    node[right, font=\scriptsize, text=black]
        {$\nabla_{\hat{\mathbf{u}}_0} \mathcal{L}_{\text{guide}}$}
    (mlp);

\node[font=\scriptsize, below=0.05cm of temb]
    {\scriptsize $\omega_j = 10000^{-2j/D}$};

\node[draw, dotted, rounded corners,
      fit=(features)(gat)(h),
      inner sep=0.35cm,
      label=above:{\textbf{Strategic Graph Encoding}}] {};

\node[draw, dotted, rounded corners,
      fit=(temb)(noise)(cat)(mlp),
      inner sep=0.35cm,
      label=below:{\textbf{Conditional MLP Diffusion}}] {};

\end{tikzpicture}
\caption{%
    System architecture for equitable utility generation.
\textbf{Strategic Graph Encoding:}
Agent feature vectors $\mathbf{x}_1, \mathbf{x}_2 \in \mathbb{R}^k$, encoding negotiation attributes (disagreement points, resource constraints, preference weights), are processed by a Graph Attention Network v2 (GATv2) with $K$ attention heads over an $n$-node directed graph, producing a shared context embedding $h \in \mathbb{R}^{d_h}$.
\textbf{Conditional Multilayer Perceptron (MLP) Diffusion:}
At each reverse diffusion step, $h$ is concatenated with the noisy state $\mathbf{u}_t \in \mathbb{R}^n$ and sinusoidal time embedding $t_{\text{emb}} \in \mathbb{R}^{d_t}$ ($\omega_j = 10000^{-2j/D}$), forming a $d_z$-dimensional input ($d_z = n + d_t + d_h$) to the MLP denoiser $s_\theta$, which predicts noise $\hat{\boldsymbol{\varepsilon}}$.
\textbf{Normative Guidance:}
In the final $t_{\text{start}}$ fraction of denoising steps, the gradient $\nabla_{\hat{\mathbf{u}}_0} \mathcal{L}_{\text{guide}}$ steers $\hat{\mathbf{u}}_0$ toward Individual Rationality and NBS using the guidance loss.
}
\label{fig:architecture}
\end{figure}
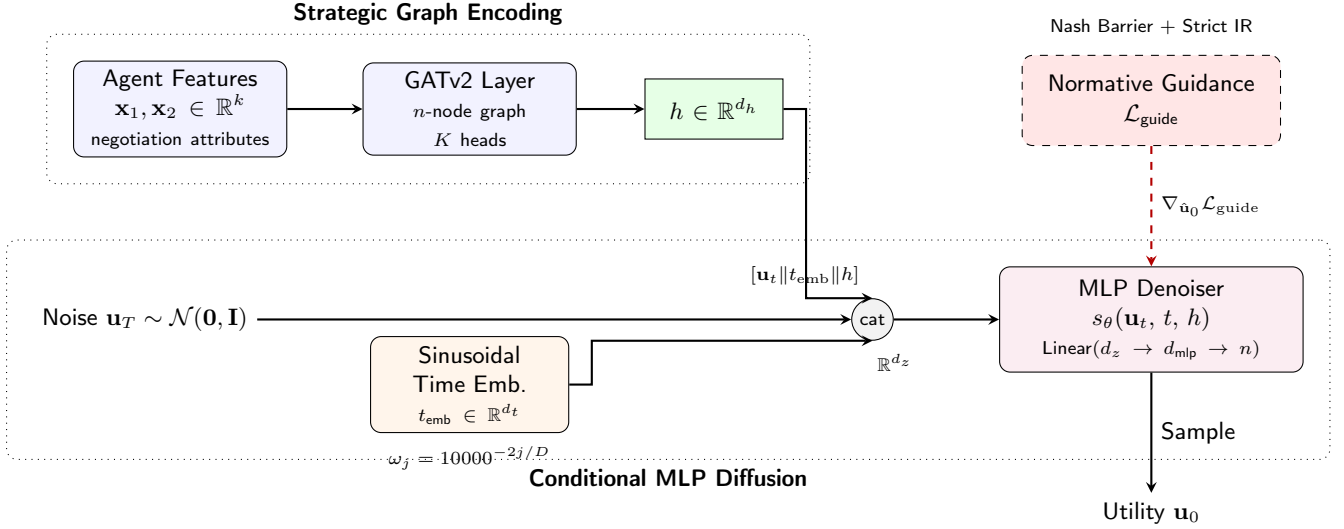

The proposed architecture (Figure~\ref{fig:architecture}) consists of three components. First, each agent’s structured negotiation state, comprising disagreement points $d_i$, resource constraints, and preference representations, is encoded as a directed graph with $n$ nodes. A Graph Attention Network v2 (GATv2)~~\cite{brody2022attentivegraphattentionnetworks} layer then processes this graph to generate a dense strategic embedding $h \in \mathbb{R}^{d_h}$ that captures inter-agent asymmetries and relational dependencies. Second, a conditional multilayer perceptron (MLP) diffusion model generates utility trajectories via denoising diffusion implicit model (DDIM)~\cite{song2020denoising} inference, conditioned on both $h$ and a sinusoidal time embedding. Finally, a normative guidance mechanism applies step-wise, gradient-based corrections to the clean utility estimate $\hat{\mathbf{u}}_0$ during the final $t_{\text{start}}$ fraction of denoising steps. 

Table~\ref{tab:notation} provides an overview of the principal mathematical notation employed throughout this formulation, together with a description of their respective conceptual roles within the model.

\begin{table}[h!]
\caption{Mathematical notation and implementation parameters.}
\label{tab:notation}
\centering
\footnotesize
\renewcommand{\arraystretch}{1.2}

\begin{tabularx}{\textwidth}{l l X  l l X}
\toprule
\textbf{Symbol} & \textbf{Code} & \textbf{Description} 
& \textbf{Symbol} & \textbf{Code} & \textbf{Description} \\
\midrule

$n$ & --- & Number of negotiating agents.
& $r$ & \texttt{radius} & Pareto frontier radius. \\

$\mathbf{u}$ & --- & Joint utility allocation vector.
& $\mathbf{u}^{*}_{\text{NBS}}$ & --- & NBS. \\

$d_i$ & --- & Disagreement point of agent $i$.
& $\mathbf{d}$ & --- & Vector of disagreement points. \\

$k$ & \texttt{feat\_dim} & Agent feature dimension.
& $K$ & \texttt{num\_heads} & Number of GATv2 heads. \\

$d_h$ & \texttt{embed\_dim} & Graph embedding dimension.
& $h$ & --- & Context embedding. \\

$T$ & \texttt{T} & Diffusion timesteps.
& $S$ & \texttt{steps} & DDIM inference steps. \\

$\beta_t$ & --- & Noise variance at step $t$.
& $\bar{\alpha}_t$ & --- & Cumulative signal retention. \\

$d_t$ & \texttt{time\_dim} & Time embedding dimension.
& $d_z$ & --- & Denoiser input dimension. \\

$d_{\text{mlp}}$ & \texttt{hidden} & MLP hidden width.
& $c_{\text{out}}$ & \texttt{scale\_out} & Output scaling factor. \\

$s_\theta$ & --- & MLP denoiser.
& $\hat{\boldsymbol{\varepsilon}}$ & --- & Predicted noise. \\

$\hat{\mathbf{u}}_0$ & --- & Clean utility estimate.
& $\hat{\mathbf{u}}_0^{\text{guided}}$ & --- & Guided estimate. \\

$c_{\max}$ & \texttt{clip\_u0} & Upper bound on $\hat{\mathbf{u}}_0$.
& $c_{\text{drift}}$ & \texttt{clip\_ut} & Clamp for latent drift. \\

$\lambda$ & \texttt{lambda\_guide} & Guidance step size.
& $t_{\text{start}}$ & \texttt{guide\_start\_frac} & Guidance start fraction. \\

$\alpha$ & \texttt{alpha\_norm} & Nash loss weight.
& $\beta$ & \texttt{beta\_ir} & IR penalty weight. \\

$\gamma$ & \texttt{gamma\_frontier} & Frontier penalty weight.
& $\delta$ & --- & IR margin constant. \\

$\epsilon$ & --- & Numerical stability constant.
& $\Pi_{\mathcal{F}}$ & --- & Projection onto feasible set. \\

\bottomrule
\end{tabularx}
\end{table}

\subsection{Strategic Graph Encoding}
\label{subsec:graph_encoding}

Conventional flat vectorial representations obscure relational asymmetries and fail to respect permutation invariance. To overcome these limitations, we model the negotiation state as a directed graph $G = (V, E)$ with $|V| = n$ nodes (representing the bilateral case when $n = 2$, with a straightforward generalization to multilateral scenarios). Each node $i \in V$ is associated with a $k$-dimensional strategic feature vector $\mathbf{x}_i \in \mathbb{R}^k$. This vector encodes domain-specific negotiation attributes, such as the disagreement point $d_i$ (forming the vector $\mathbf{d}$), resource constraints, and preference orientations. Directed edges represent the dyadic relational structure and constraints between agents.

To extract a unified conditioning signal from this structure, we employ a GATv2 layer. Unlike standard graph convolutions, GATv2 uses a dynamic attention mechanism that weights the influence of agent $j$ on agent $i$ based on their joint strategic state. For each attention head, the unnormalised attention score is computed as:
\begin{equation}
    e_{ij} = \mathbf{a}^{\top}
    \operatorname{LeakyReLU}\!\left(
        \mathbf{W}\,[\mathbf{x}_i \,\|\, \mathbf{x}_j]
    \right)
    \label{eq:gat_score}
\end{equation}
where $\mathbf{W} \in \mathbb{R}^{d_h \times 2k}$ and $\mathbf{a} \in \mathbb{R}^{d_h}$ are learnable parameters, and $\|$ denotes concatenation. These scores are then normalised across the neighbourhood $\mathcal{N}(i)$ via a softmax function:
\begin{equation}
    \alpha_{ij} =
    \frac{\exp(e_{ij})}
         {\displaystyle\sum_{\ell \in \mathcal{N}(i)} \exp(e_{i\ell})}
    \label{eq:gat_alpha}
\end{equation}
The updated node representation for head $m$ is obtained by aggregating the weighted features:
\begin{equation}
    \mathbf{h}_i^{(m)} =
    \sigma_{\text{act}}\!\left(
        \sum_{j \in \mathcal{N}(i)}
        \alpha_{ij}^{(m)}\, \mathbf{W}^{(m)} \mathbf{x}_j
    \right)
    \label{eq:gat_node}
\end{equation}
where $\sigma_{\text{act}}$ denotes a non-linear activation function. Using $K$ parallel attention heads, the per-node outputs are concatenated and passed through a linear projection, followed by layer normalisation, to yield $\mathbf{h}_i \in \mathbb{R}^{d_h}$. Finally, a global mean-pooling readout produces the unified conditioning embedding:
\begin{equation}
    h = \frac{1}{|V|} \sum_{i \in V} \mathbf{h}_i
    \;\in\; \mathbb{R}^{d_h}
    \label{eq:pooling}
\end{equation}
This formulation ensures that $h$ is permutation-invariant while remaining sensitive to the asymmetric strategic constraints of the negotiating agents. This embedding is computed once per negotiation context and is held fixed, providing a stable geometric prior throughout all subsequent diffusion inference steps.

\subsection{Conditional MLP Diffusion Process}
\label{subsec:diffusion}

Having distilled the structural asymmetries of the negotiation into the strategic context embedding $h$, we formulate the generative engine as a conditional Denoising Diffusion Probabilistic Model (DDPM). The objective of this module is to synthesise a joint utility vector $\mathbf{u}_0 \in \mathbb{R}^n$ that represents an equitable $n$-party allocation.

The formulation begins with the forward diffusion process. A clean utility sample $\mathbf{u}_0$ is progressively corrupted into isotropic Gaussian noise over $T$ discrete timesteps via a predefined variance schedule $\{\beta_t\}_{t=1}^{T}$:
\begin{equation}
    q(\mathbf{u}_t \mid \mathbf{u}_0)
    = \mathcal{N}\!\left(
        \mathbf{u}_t;\;
        \sqrt{\bar\alpha_t}\,\mathbf{u}_0,\;
        (1 - \bar\alpha_t)\,\mathbf{I}
      \right),
    \qquad
    \bar\alpha_t = \prod_{s=1}^{t}(1 - \beta_s)
    \label{eq:forward}
\end{equation}
This allows for the direct sampling of the noisy state at any arbitrary timestep $t$ as $\mathbf{u}_t = \sqrt{\bar\alpha_t}\,\mathbf{u}_0 + \sqrt{1 - \bar\alpha_t}\,\boldsymbol{\varepsilon}$, where $\boldsymbol{\varepsilon} \sim \mathcal{N}(\mathbf{0}, \mathbf{I})$, thereby bypassing sequential simulation. 

To invert this corruption, the reverse process is parameterised by a denoiser $s_\theta(\mathbf{u}_t, t, h)$ that predicts the injected noise $\hat{\boldsymbol{\varepsilon}}$. We adopt a Multilayer Perceptron (MLP) architecture for this task. 

At each inference step, the MLP receives a $d_z$-dimensional input vector formed by concatenating the noisy utility state $\mathbf{u}_t \in \mathbb{R}^n$, a sinusoidal time embedding $t_{\text{emb}} \in \mathbb{R}^{d_t}$, and the fixed graph context $h \in \mathbb{R}^{d_h}$:
\begin{equation}
    \mathbf{z}
    = [\,\mathbf{u}_t \;\|\; t_{\text{emb}} \;\|\; h\,]
    \;\in\; \mathbb{R}^{d_z}, \qquad
    d_z = n + d_t + d_h
    \label{eq:concat}
\end{equation}
The sinusoidal embedding encodes the normalised timestep $\tilde{t} = t/T \in [0,1]$ using frequency components $\omega_j = 10000^{-2j/D}$ for $j = 0, \ldots, D/2 - 1$, followed by a two-layer projection. The denoiser maps $\mathbf{z}$ to the predicted noise via a stack of fully-connected layers, interspersed with Layer Normalisation (LN) and SiLU activations:
\begin{equation}
    s_\theta:\;
    \mathbb{R}^{d_z}
    \xrightarrow{\;\text{Linear}\;}
    \xrightarrow{\;\text{LN} + \text{SiLU}\;}
    \xrightarrow{\;\text{Linear}\;}
    \xrightarrow{\;\text{LN} + \text{SiLU}\;}
    \xrightarrow{\;\text{Linear}\;}
    \hat{\boldsymbol{\varepsilon}} \in \mathbb{R}^n
    \label{eq:mlp}
\end{equation}
To ensure numerical stability, the final layer output is scaled by a small constant $c_{\text{out}} \ll 1$, which prevents the clean estimate $\hat{\mathbf{u}}_0$ from exploding at high noise levels where $\sqrt{\bar\alpha_t} \approx 0$. 

During inference, the model generates utility vectors using $S$ deterministic DDIM steps. At each step, a clean utility estimate is first recovered from the noisy state:
\begin{equation}
    \hat{\mathbf{u}}_0
    = \operatorname{clip}\!\left(
        \frac{
            \mathbf{u}_t
            - \sqrt{1-\bar\alpha_t}\;\hat{\boldsymbol\varepsilon}
        }{\sqrt{\bar\alpha_t}},\;
        0,\; c_{\max}
      \right)
    \label{eq:u0hat}
\end{equation}
where $c_{\max}$ is an upper clamp bound that stabilises the estimate. Once recovered, $\hat{\mathbf{u}}_0$ is subjected to normative guidance (detailed in Section~\ref{subsec:guidance}) to yield $\hat{\mathbf{u}}_0^{\text{guided}}$. The DDIM update then propagates this guided estimate to the subsequent step:
\begin{equation}
    \mathbf{u}_{t-1}
    = \sqrt{\bar\alpha_{t-1}}\;\hat{\mathbf{u}}_0^{\text{guided}}
    + \sqrt{1-\bar\alpha_{t-1}}\;\hat{\boldsymbol\varepsilon}
    \label{eq:ddim}
\end{equation}
To prevent cumulative latent drift from compounding over multiple guided steps, the intermediate state is soft-clamped: $\mathbf{u}_{t-1} \leftarrow \operatorname{clip}(\mathbf{u}_{t-1},\, -c_{\text{drift}},\, c_{\text{drift}})$. At the terminal step, rather than returning the DDIM output (which collapses to $\hat{\boldsymbol\varepsilon}$ when $\bar\alpha_{\text{next}} = 0$), the model directly outputs the final guided estimate $\hat{\mathbf{u}}_0^{\text{guided}}$.

\paragraph{Training Objective.}
Model optimisation proceeds in two phases. In Phase~1, the denoiser is trained using a standard denoising score-matching loss~\cite{ho2020denoising} to accurately predict the injected noise:
\begin{equation}
    \mathcal{L}_{\text{MSE}}
    = \mathbb{E}_{t,\,\mathbf{u}_0,\,\boldsymbol\varepsilon}
      \!\left[\,
        \left\|
          \boldsymbol\varepsilon
          - s_\theta\!\left(
              \sqrt{\bar\alpha_t}\,\mathbf{u}_0
              + \sqrt{1-\bar\alpha_t}\,\boldsymbol\varepsilon,\;
              t,\; h
            \right)
        \right\|^2
      \right]
    \label{eq:mse}
\end{equation}
In Phase~2, a normative regulariser is added:
$\mathcal{L} = \mathcal{L}_{\text{MSE}} + \mathcal{L}_{\text{guide}}$,
where $\mathcal{L}_{\text{guide}}$ is applied to the clipped estimate
$\hat{\mathbf{u}}_0$ and is formally defined in
Section~\ref{subsec:guidance} (Eq.~\ref{eq:lguid}).

\subsection{Normative Guidance via Differentiable Nash Bargaining}
\label{subsec:guidance}

As established in Section~\ref{subsec:diffusion}, relying solely on standard score-matching risks replicating the historical biases embedded in the training data. To counteract this and steer the generative process toward equitable outcomes in Non-Transferable Utility (NTU) environments, we introduce a differentiable guidance loss, $\mathcal{L}_{\text{guide}}$. This function serves a dual purpose: it acts as a normative regularizer during Phase 2 training and as a dynamic gradient guide during inference. To prevent premature constraint interference, this guidance is applied dynamically only during the final $t_{\text{start}}$ fraction of reverse denoising steps (i.e., when $t/T < t_{\text{start}}$).

Let $d_i$ denote the disagreement point for agent $i$. The three-term guidance loss for the $n$-party case is formulated as:

\begin{equation}
\begin{aligned}
\mathcal{L}_{\text{guide}}(\mathbf{u}, \mathbf{d})
&=
\underbrace{
-\alpha \sum_{i=1}^{n}
\log\!\Big(
\operatorname{softplus}(u_i - d_i) + \epsilon
\Big)
}_{\mathcal{L}_N:\;\text{Nash efficiency}}
+
\underbrace{
\beta \sum_{i=1}^{n}
\operatorname{softplus}(d_i - u_i + \delta)
}_{\mathcal{L}_{IR}:\;\text{strict IR}}
\\[4pt]
&\quad +
\underbrace{
\gamma \cdot \operatorname{softplus}(\|\mathbf{u}\|^2 - r^2)
}_{\mathcal{L}_F:\;\text{frontier adherence}}
\end{aligned}
\label{eq:lguid}
\end{equation}

where $\operatorname{softplus}(x) = \log(1 + e^x)$ is utilized in place of hard maximum operators to ensure smooth, everywhere-nonzero gradients. Within this formulation, $\epsilon > 0$ provides numerical stability; $\delta > 0$ acts as a strictness margin to enforce $u_i \geq d_i$ with a conservative buffer; and $r > 0$ defines the Pareto frontier radius, bounding the feasible set $\mathcal{F} = \{\mathbf{u} \geq \mathbf{0},\; \|\mathbf{u}\| \leq r\}$.

At each guided inference step, the gradient $g = \nabla_{\hat{\mathbf{u}}_0} \mathcal{L}_{\text{guide}}$ is computed via automatic differentiation. To scale the update proportionally to the current signal-to-noise ratio, we employ an adaptive step size $w = \lambda\sqrt{\bar\alpha_t}$. This provides strong, necessary corrections at low noise levels while avoiding destabilizing interference when the noise variance is high. The guided utility estimate is computed as:

\begin{equation}
    \hat{\mathbf{u}}_0^{\text{guided}}
    = \Pi_{\mathcal{F}}\!\left(
        \hat{\mathbf{u}}_0
        - \frac{w\, g}{\|g\| + \epsilon}
      \right)
    \label{eq:guidance_step}
\end{equation}

where $\Pi_{\mathcal{F}}$ is the Euclidean projection onto the feasible set, executed by first clamping negatives ($\mathbf{u} \leftarrow \max(\mathbf{u}, \mathbf{0})$) and subsequently projecting onto the radius ($\mathbf{u} \leftarrow \mathbf{u} \cdot r / \max(\|\mathbf{u}\|, r)$). Gradient normalisation is applied prior to projection to prevent the three disparate terms in $\mathcal{L}_{\text{guide}}$ from producing disproportionately large updates when the state operates far from the Pareto frontier.

Crucially, the Softplus-based formulation is adopted over hard $\max(0, \cdot)$ and $\log(\max(\epsilon, \cdot))$ operators for two optimization-critical reasons. First, its derivative, $\partial \operatorname{softplus}(x) / \partial x = \sigma(x)$, is bounded within $(0,1)$ for all $x$ (where $\sigma$ here denotes the standard logistic sigmoid), preventing gradient explosions when $\hat{\mathbf{u}}_0$ violates the Individual Rationality boundary by a large margin. Second, because the gradient remains nonzero everywhere, the guidance mechanism continues to provide a directional corrective signal even in deeply infeasible regions, whereas a hard maximum operator would produce a zero gradient and stall the recovery process.

\subsection{Theoretical Guarantees and Boundary Cases}
\label{subsec:theory}

We establish the theoretical properties of the guided reverse diffusion
process. Let $\mathbf{d} = [d_1, \ldots, d_n]^{\top}$ denote the vector of
disagreement points. The Individually Rational feasible region is
$\mathcal{F} = \{\mathbf{u} \in \mathbb{R}^n \mid u_i \geq d_i,\;
\forall i\}$. The reverse-time probability flow ODE governing the guided
trajectory is:
\begin{equation}
    d\mathbf{u}_t
    = \left[
        \mathbf{f}(\mathbf{u}_t, t)
        - \tfrac{1}{2}\,g^2(t)
          \Big(
            \nabla_{\mathbf{u}} \log p_t(\mathbf{u}_t)
            - \nabla_{\mathbf{u}}
              \mathcal{L}_{\text{guide}}(\mathbf{u}_t)
          \Big)
      \right] dt
    \label{eq:ode}
\end{equation}

We invoke the following standard assumptions:

\begin{assumption}[Manifold Geometry]
The diffusion model implicitly learns a data manifold $\mathcal{M}$ whose
upper-right boundary constitutes a strictly concave, continuous Pareto
frontier $\partial\mathcal{M}^+$.
\end{assumption}

\begin{assumption}[Score Boundedness]
The unconditional score function $\nabla_{\mathbf{u}} \log p_t(\mathbf{u}_t)$
is $L$-Lipschitz continuous and bounded:
$\|\nabla_{\mathbf{u}} \log p_t(\mathbf{u}_t)\| \leq M$ for some $M > 0$.
\end{assumption}

\begin{assumption}[Nash Concavity]
The Nash product
$\mathcal{N}(\mathbf{u}) = \prod_{i=1}^{n}(u_i - d_i)$ is strictly
log-concave over the convex set $\mathcal{M} \cap \mathcal{F}$.
\end{assumption}

Assumptions~1 and~3 hold for divisible resources with a convex feasible set.
For non-convex domains, the forward noising process acts as a Gaussian
smoothing operator that mathematically induces a convex manifold
$\mathcal{M}_\theta$ at higher noise scales, allowing the reverse process to
bypass local optima before annealing into the true frontier.

\begin{lemma}[Finite-Time Convergence to IR]
For any initial sample $\mathbf{u}_T \notin \mathcal{F}$, if the IR penalty
weight satisfies $\beta > M + \sup_{\mathbf{u} \notin \mathcal{F},\, t} \|\mathbf{f}(\mathbf{u}, t)\|$, the guided trajectory enters $\mathcal{F}$ in finite time.
\end{lemma}
\begin{proof}
Define the Lyapunov function
$V(\mathbf{u}) = \tfrac{1}{2}\sum_{i=1}^{n}
\big(\max(0,\, d_i - u_i)\big)^2$.
Outside $\mathcal{F}$, the IR gradient of $\mathcal{L}_{\text{guide}}$
dominates the corrective update. Under the probability flow ODE (Eq.~\ref{eq:ode}), because the penalty weight $\beta$ strictly bounds both the maximal score magnitude $M$ and the drift term $\mathbf{f}(\mathbf{u}, t)$, the guidance term dominates the trajectory:
$\nabla_{\mathbf{u}} \mathcal{L}_{\text{guide}} \approx
-\beta\sum_{\{i:\,u_i < d_i\}} \mathbf{e}_i$.
It follows that the continuous-time derivative satisfies $\dot{V}(\mathbf{u}) \leq -c\sqrt{V(\mathbf{u})}$ for
some constant $c > 0$. By Lyapunov's direct method, the trajectory reaches $V = 0$ (and thus enters $\mathcal{F}$) in finite time.
\end{proof}

\begin{theorem}[Asymptotic Convergence to the NBS]
Given $\mathbf{u}_t \in \mathcal{F}$, and assuming the learned score perfectly matches the true marginal score ($\nabla_{\mathbf{u}} \log p_t(\mathbf{u}_t) = s_\theta(\mathbf{u}_t, t)$), as the noise schedule $g(t) \to 0$ and $\alpha \to 1$, the continuous-time dynamics converge to the unique NBS $\mathbf{u}^*_{\text{NBS}} \in \partial\mathcal{M}^+$.
\end{theorem}
\begin{proof}
Inside the feasible region $\mathcal{F}$, the IR penalty evaluates to zero, and the Nash term dominates the guidance gradient:
$\nabla_{\mathbf{u}} \mathcal{L}_{\text{guide}} =
-\alpha\nabla_{\mathbf{u}} \log \mathcal{N}(\mathbf{u})$.
Under the assumption of perfect score matching, the unconditional score smoothly projects the state onto the tangent space of the data manifold $\mathcal{M}$, while the guidance gradient pushes the state toward higher Nash product values along this manifold. By Assumptions~1 and~3, maximizing a strictly log-concave function over a convex set yields a unique global maximum. Therefore, as the noise variance collapses ($t \to 0$), the state converges deterministically:
$\lim_{t \to 0} \mathbf{u}_t
= \arg\max_{\mathbf{u} \in \partial\mathcal{M}^+ \cap \mathcal{F}}
\prod_{i=1}^{n}(u_i - d_i)
= \mathbf{u}^*_{\text{NBS}}$.
\end{proof}

The final utility $\hat{\mathbf{u}}$ approximates the NBS relative to the
model's learned distribution $\mathcal{M}_\theta$. In practice, the distance between the generated sample $\hat{\mathbf{u}}$ and the true $\mathbf{u}^*_{\text{NBS}}$ is bounded by the $L_2$ score-matching error of the neural network, providing a stochastically smoothed approximation well-suited to noisy empirical datasets.

\subsection{Practical Inference and Counterfactual Evaluation}
\label{subsec:inference}

At runtime, the utility generation process is initialised with pure Gaussian
noise $\mathbf{u}_T \sim \mathcal{N}(\mathbf{0}, \mathbf{I})$. This noise is
iteratively refined over $S$ DDIM steps, conditioned on the graph embedding
$h$ and steered by $\nabla_{\hat{\mathbf{u}}_0} \mathcal{L}_{\text{guide}}$
in the final $t_{\text{start}}$ fraction of steps. The complete inference
procedure is summarised in Algorithm~\ref{alg:inference}.

\begin{algorithm}[htbp]
\caption{Guided DDIM Inference}
\label{alg:inference}
\begin{algorithmic}[1]
\Require Agent features $\mathbf{x}_1, \ldots, \mathbf{x}_n$; disagreement vector $\mathbf{d}$;
         parameters $\lambda$, $t_{\text{start}}$, $\alpha$, $\beta$,
         $\gamma$, $r$; number of sampling steps $S$
\Ensure  $\hat{\mathbf{u}} \in \mathbb{R}^n$ (equitable utility allocation)
\State $h \leftarrow \text{GATv2}(\mathbf{x}_1, \ldots, \mathbf{x}_n)$
    \Comment{Context embedding $\in \mathbb{R}^{d_h}$}
\State $\mathbf{u}_T \sim \mathcal{N}(\mathbf{0}, \mathbf{I})$
\For{$i = 1, \ldots, S$}
    \State $t \leftarrow \text{timestep corresponding to step } i$
    \State $\hat{\boldsymbol\varepsilon}
           \leftarrow s_\theta(\mathbf{u}_t, t, h)$
    \State $\hat{\mathbf{u}}_0 \leftarrow
           \operatorname{clip}\!\left(
               \dfrac{\mathbf{u}_t
                      - \sqrt{1-\bar\alpha_t}\,\hat{\boldsymbol\varepsilon}}
                     {\sqrt{\bar\alpha_t}},\;
               0,\; c_{\max}
           \right)$
    \If{$t/T < t_{\text{start}}$}
        \Comment{Guidance window}
        \State $w \leftarrow \lambda\sqrt{\bar\alpha_t}$
        \State $g \leftarrow
               \nabla_{\hat{\mathbf{u}}_0}
               \mathcal{L}_{\text{guide}}(\hat{\mathbf{u}}_0, \mathbf{d})$
        \State $\hat{\mathbf{u}}_0 \leftarrow
               \Pi_{\mathcal{F}}\!\left(
                   \hat{\mathbf{u}}_0
                   - w\,g\,/\,(\|g\| + \epsilon)
               \right)$
    \EndIf
    \If{$i = S$}
        \State \Return $\hat{\mathbf{u}}_0$
            \Comment{Terminal step: return clean estimate directly}
    \EndIf
    \State $\mathbf{u}_{t-1} \leftarrow
           \sqrt{\bar\alpha_{t-1}}\,\hat{\mathbf{u}}_0
           + \sqrt{1-\bar\alpha_{t-1}}\,\hat{\boldsymbol\varepsilon}$
    \State $\mathbf{u}_{t-1} \leftarrow
           \operatorname{clip}(\mathbf{u}_{t-1},\;
           -c_{\text{drift}},\; c_{\text{drift}})$
           \Comment{Prevent cumulative drift}
\EndFor
\end{algorithmic}
\end{algorithm}

Beyond generating a single outcome, the stochastic initialisation supports
ensemble-based evaluation and diversity quantification by sampling multiple
trajectories from the same context. Furthermore, the generative architecture
naturally enables counterfactual fairness analysis: targeted interventions on
the input features, such as equalising the disagreement points $d_i$ of
asymmetrically situated agents, allow quantitative evaluation of the equity
impact of structural imbalances, without requiring access to counterfactual
ground-truth outcomes.

\section{Experimental Setup }
\label{sec:experimental_setup}

To empirically validate the guided graph diffusion framework for multi-issue bilateral negotiations, the experimental pipeline is designed as a coherent, sequential workflow. It begins with carefully curated data infrastructure, proceeds through model architecture instantiation and phased training, incorporates systematic hyperparameter optimization, and concludes with a multi-dimensional evaluation protocol. Every negotiation instance is normalized to a strictly bilateral continuous utility space $\mathbf{u} \in [0,1]^2$, which enables consistent comparison between theoretically optimal bargaining outcomes and both synthetic baselines and human-generated behavior.

\subsection{Datasets}
\label{subsec:datasets}

The empirical foundation rests on three complementary datasets that together span analytical exactness and real-world complexity (see Table~\ref{tab:datasets} for a summary of their structural properties and splits). The synthetic NTU dataset comprises $25{,}000$ dyads featuring strictly concave Pareto frontiers and asymmetrically sampled disagreement points drawn from $\mathcal{U}(0.05, 0.4)$. Because the disagreement points and frontier shapes are fully known,  NBS can be computed numerically using Sequential Least Squares Programming (SLSQP)~\cite{sw36348}, providing an unambiguous analytical reference. 

To assess how well the model generalizes to human bargaining behavior, two established real-world corpora are included. The CaSiNo corpus captures human-to-human multi-issue negotiations over campsite items (food, water, firewood) in which participants hold heterogeneous but publicly revealed priority weights. The Deal-or-No-Deal corpus, by contrast, features hidden private valuations over splits of books, hats, and balls, introducing realistic information asymmetry. In both human datasets, utilities are normalized to the unit square and no-deal outcomes are excluded to focus on reached agreements. Across all three datasets, each agent is represented by the same compact 3-dimensional feature vector $[d_i, \text{normalized budget}, \text{priority weight}]$, ensuring architectural uniformity while respecting domain-specific preference structures.

\begin{table}[htbp]
\centering
\caption{Dataset Infrastructure and Structural Properties}
\label{tab:datasets}
\renewcommand{\arraystretch}{1.3}
\resizebox{\textwidth}{!}{%
\begin{tabular}{llccp{6.5cm}}
\toprule
\textbf{Dataset}
  & \textbf{Split}
  & \textbf{Disagreement $d_i$}
  & \textbf{Feat.\ dim $k$}
  & \textbf{Structural Characteristics} \\
\midrule
Synthetic NTU
  & 20k / 2.5k / 2.5k
  & $\mathcal{U}(0.05,\,0.4)$
  & 3
  & $25{,}000$ dyads with strictly concave NTU Pareto frontiers;
    disagreement points sampled asymmetrically per agent;
    exact NBS computable via SLSQP. \\
CaSiNo corpus \cite{chawla-etal-2021-casino}
  & 80\% / 10\% / 10\%
  & $--$
  & 3
  & Human-to-human multi-issue bargaining over campsite items
    (food, water, firewood) with varying priority weights;
    utilities normalized to $[0,1]^2$. \\
Deal or No Deal \cite{lewis2017dealornodeal}
  & 80\% / 10\% / 10\%
  & $--$
  & 3
  & Agent negotiations over item splits (books, hats, balls)
    with hidden private valuations; no-deal outcomes excluded. \\
\bottomrule
\multicolumn{5}{l}{\footnotesize
  Feature vector $\mathbf{x}_i \in \mathbb{R}^k$:
  $[d_i,\; \text{normalized budget},\; \text{priority weight}]$
  for all three datasets.}
\end{tabular}%
}
\end{table}

\subsection{Model Architecture and Training Configuration}
\label{subsec:architecture}

The model architecture strategically decouples relational context encoding from the generative denoising process. All fixed architectural choices and numerical bounds are reported in Table~\ref{tab:architecture}. A single GATv2 layer with four attention heads processes a directed two-node bargaining graph, after which global mean-pooling yields a fixed 64-dimensional context vector that conditions the entire reverse diffusion trajectory. Given the minimal structural complexity of bilateral negotiations, deep graph message passing provides limited benefit; therefore, the noise-prediction network $s_\theta$ is implemented as a time-conditioned MLP. This MLP directly concatenates the noisy utility state $\mathbf{u}_t \in \mathbb{R}^2$, a 32-dimensional sinusoidal time embedding, and the pooled context vector into a 98-dimensional input, which is then processed through two hidden layers of width 256 using LayerNorm and SiLU activations. 

Training unfolds over 30 epochs in two distinct phases. During the first 15 epochs (Phase 1), the model learns the empirical data manifold via standard denoising score matching with a linear variance schedule ($\beta_t$ from $10^{-4}$ to 0.02) over 1000 forward timesteps and AdamW optimization with cosine annealing. In the second phase (epochs 15--29), the normative guidance loss $\mathcal{L}_{\text{guide}}$ is gradually introduced. To avoid destabilization, the individual rationality penalty coefficient $\beta$ is linearly annealed from 10 to 50, progressively strengthening adherence to game-theoretic axioms without overwhelming early training dynamics.

\begin{table}[htbp]
\centering
\caption{Fixed architectural hyperparameters and numerical bounds.}
\label{tab:architecture}
\small
\renewcommand{\arraystretch}{1.15}
\begin{tabular}{lc | lc}
\toprule
\textbf{Parameter} & \textbf{Value} & \textbf{Parameter} & \textbf{Value} \\
\midrule
\multicolumn{4}{c}{\textit{Architecture \& Diffusion Schedule}} \\
\midrule
Agents ($n$) & 2 & Total Timesteps ($T$) & 1000 \\
GATv2 Heads ($K$) & 4 & DDIM Steps ($S$) & 50 \\
Graph Emb. ($d_h$) & 64 & Initial Var. ($\beta_1$) & $10^{-4}$ \\
Time Emb. ($d_t$) & 32 & Final Var. ($\beta_T$) & 0.02 \\
MLP Width ($d_{\text{mlp}}$) & 256 & Input Dim. ($d_z$) & 98 \\
\midrule
\multicolumn{4}{c}{\textit{Numerical Bounds \& Training Configuration}} \\
\midrule
Clean Bound ($c_{\max}$) & 1.2 & Stability Const. ($\epsilon$) & $10^{-6}$ \\
Drift Clamp ($c_{\text{drift}}$) & 1.5 & IR Margin ($\delta$) & 0.05 \\
Output Scale ($c_{\text{out}}$) & 0.1 & Pareto Radius ($r$) & 1.0 \\
Batch Size & 256 & Total Epochs & 30 \\
Optimizer & AdamW & Phase 1 / Phase 2 & 15 / 15 epochs \\
\bottomrule
\end{tabular}
\end{table}

\subsection{Hyperparameter Search}
\label{subsec:search}

Because the impact of classifier-free-style guidance is known to be highly sensitive to scale, timing, and relative weighting, a comprehensive multi-dimensional grid search was performed over the key normative guidance hyperparameters rather than adopting arbitrary fixed values. The complete search space, including guidance strength $\lambda$, activation timestep $t_{\text{start}}$, Nash-product multiplier $\alpha$, IR penalty $\beta$, frontier penalty $\gamma$, and number of DDIM steps $S$, is detailed in Table~\ref{tab:search_space}. An extensive grid of configurations was evaluated on every test split using the full suite of metrics, and the best-performing parameter set for each dataset was selected via a composite objective that balances Nash efficiency, individual rationality compliance, and frontier proximity.

\begin{table}[htbp]
\centering
\caption{Multidimensional Search Space for Normative Guidance Parameters}
\label{tab:search_space}
\renewcommand{\arraystretch}{1.2}
\small
\begin{tabular}{llcc}
\toprule
\textbf{Parameter}
  & \textbf{Symbol}
  & \textbf{Search Domain}
  & \textbf{Grid Steps} \\
\midrule
Guidance step size & $\lambda$ & $[0.005,\; 0.35]$ & 10 \\
Guidance activation & $t_{\text{start}}$ & $[0.10,\; 0.70]$ & 8 \\
Nash multiplier & $\alpha$ & $[5.0,\; 300.0]$ & 10 \\
IR penalty & $\beta$ & $[0.5,\; 50.0]$ & 10 \\
Frontier penalty & $\gamma$ & $[0.5,\; 80.0]$ & 10 \\
DDIM inference steps & $S$ & $\{10,\; 15,\; 25,\; 50\}$ & 4 \\
\bottomrule
\end{tabular}
\end{table}

\subsection{Evaluation Metrics}
\label{subsec:metrics}

System performance is quantified through four complementary metrics that jointly probe normative validity, efficiency, and geometric fidelity. IR compliance measures the expected fraction of generated samples that strictly dominate the disagreement point for both agents simultaneously: 

$$
    \text{IR Compliance} = \mathbb{E}[\mathbf{1}[u_i \geq d_i \;\forall i]] 
$$

The Nash product serves as the primary proxy for joint surplus and fairness, computed as $\mathbb{E}[\prod_i (u_i - d_i)]$. 

Nash efficiency normalizes the achieved Nash product against the exact SLSQP-computed NBS of each test instance:

$$
    \text{Nash Efficiency} = \frac{\mathbb{E}[\text{Generated Nash Product}]}{\mathbb{E}[\text{NBS Nash Product}]}
$$

enabling fair cross-dataset comparison despite differing disagreement geometries. 

Finally, mean $L_2$ frontier distance $(\mathcal{D}_F)$ quantifies geometric accuracy by measuring the average Euclidean projection distance from generated points to the theoretical Pareto frontier arc (independent of the Nash objective):

$$
    \mathcal{D}_F = \mathbb{E}\!\left[ \left\| \mathbf{u} - \frac{r\,\mathbf{u}}{\|\mathbf{u}\|} \right\|_2 \right]
$$

To assess the statistical significance of performance improvements attributable to normative guidance, Wilcoxon signed-rank tests are conducted on the per-instance Nash-product distributions derived from guided versus unguided models. This test constitutes a robust, non-parametric procedure that is well suited to the characteristically skewed distributions observed for optimal bargaining outcomes.

Taken together, this experimental framework, grounded in the dataset specifications reported in Table~\ref{tab:datasets}, instantiated through the architectural and training configurations in Table~\ref{tab:architecture}, systematically tuned via the sensitivity analysis summarized in Table~\ref{tab:search_space}, and evaluated using a comprehensive and balanced suite of performance metrics—provides a methodologically rigorous framework for linking axiomatic bargaining theory and the generative capabilities of guided diffusion models in realistic, multi-issue negotiation environments.

\section{Results and Analysis}
\label{sec:results}

The efficacy of the normatively guided graph diffusion framework is evaluated
in three phases. First, the sensitivity of the guidance hyperparameters is
analyzed to establish optimal configurations for each domain. Second,
quantitative improvements in axiomatic compliance and joint surplus
maximization achieved by the guided model over the unguided baseline are
assessed. Finally, the latent trajectory dynamics are examined to understand
how continuous gradient steering constrains the diffusion process to satisfy
strict game-theoretic constraints.

\subsection{Sensitivity Analysis and Hyperparameter Optimization}
\label{subsec:sensitivity}

To ensure stable convergence of the guided diffusion chain while jointly
enforcing IR and NBS alignment, a multidimensional
sensitivity analysis was conducted across six guidance hyperparameters: the
normalized step size $\lambda$, the guidance activation threshold $t_{\text{start}}$,
the Nash multiplier $\alpha$, the IR penalty $\beta$, the frontier penalty
$\gamma$, and the number of DDIM inference steps $S$. The explicit correspondence 
between these mathematical symbols and the programmatic variable names displayed 
in the subsequent sensitivity visualizations is documented in Table~\ref{tab:notation}. A formal sensitivity score was computed for each parameter, defined as the range of the evaluation metric divided by its mean, averaged across all four target metrics (IR Compliance $\uparrow$, Nash Product $\uparrow$, Nash Efficiency $\uparrow$, Frontier Distance $\downarrow$).

\begin{figure}[h!]
    \centering
    \includegraphics[width=\textwidth]{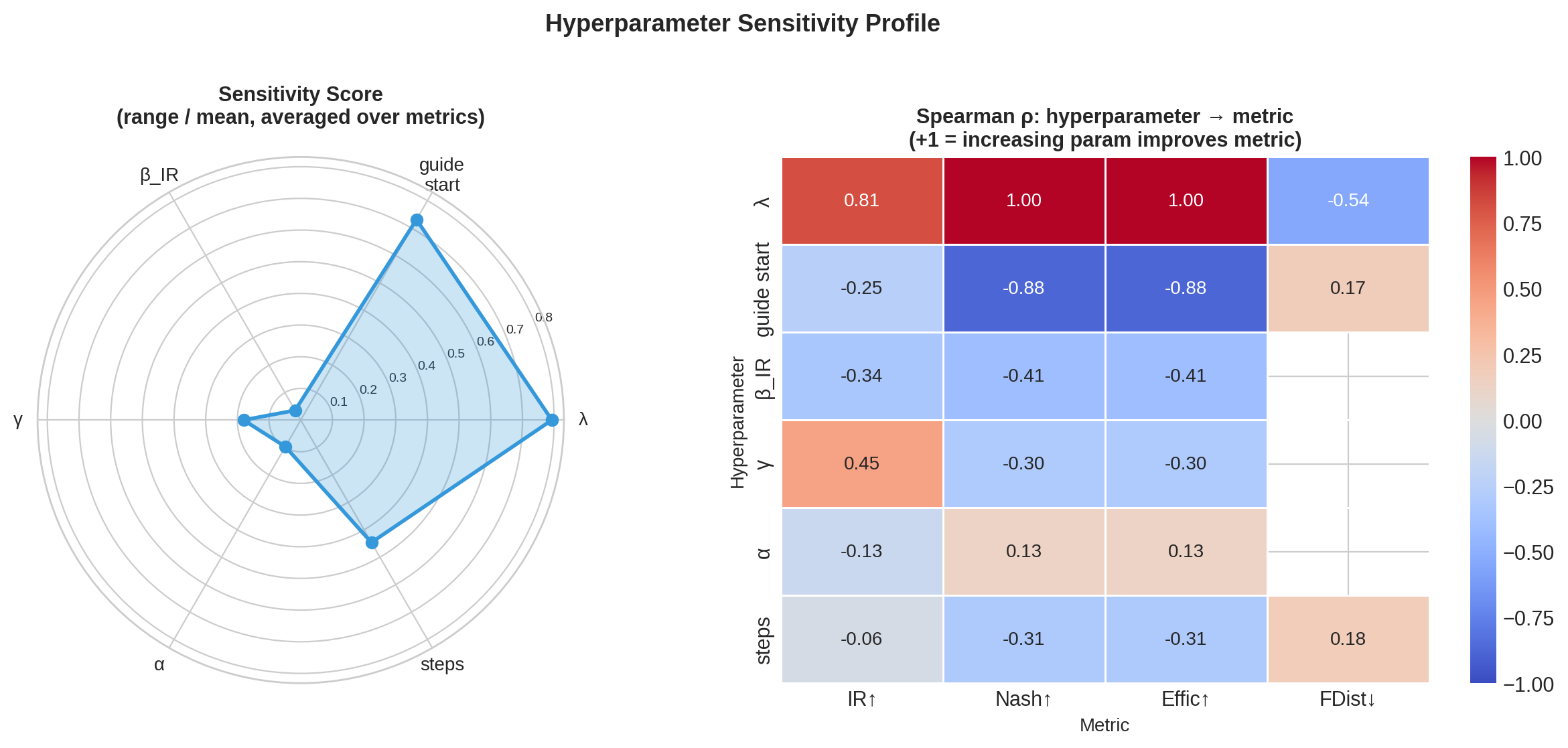}
    \caption{%
        Hyperparameter sensitivity profile for the Synthetic NTU dataset.
        \textbf{Left:} Spider chart ranking parameters by sensitivity score
        (range/mean averaged over four metrics), identifying $t_{\text{start}}$
        and $\lambda$ as the dominant factors (${\approx}0.80$ and
        ${\approx}0.70$ respectively), while $\alpha$ and $\gamma$ are
        effectively saturated (${\leq}0.20$).
        \textbf{Right:} Spearman $\rho$ matrix quantifying directionality.
        $\lambda$ exhibits perfect monotone correlation with Nash Product
        ($\rho=+1.00$) and Nash Efficiency ($\rho=+1.00$), while
        $t_{\text{start}}$ is strongly negatively correlated ($\rho=-0.88$),
        confirming that earlier, stronger guidance is the primary driver of
        Nash-optimal allocation. The $\beta$ row exhibits a counterintuitive
        negative correlation ($\rho \approx -0.34$ with IR), attributable
        to gradient overshooting at excessive penalty magnitudes.
        Symbol--code name correspondence follows Table~\ref{tab:notation}.%
    }
    \label{fig:sensitivity_radar}
\end{figure}

As quantified in Figure~\ref{fig:sensitivity_radar} (left), $t_{\text{start}}$
and $\lambda$ emerge as the dominant parameters, yielding sensitivity scores of
${\approx}0.80$ and ${\approx}0.70$ respectively. Conversely, $\alpha$ exhibits
the lowest sensitivity (${\approx}0.15$), indicating saturation within its
tested range and implying that search effort should be concentrated on $\lambda$
and $t_{\text{start}}$.

The Spearman $\rho$ matrix (Figure~\ref{fig:sensitivity_radar}, right)
clarifies the directionality of these effects. The $\lambda$ row reveals the
strongest signal: perfect monotone correlations with Nash Product and Nash
Efficiency ($\rho = +1.00$), alongside a strong positive correlation with IR
Compliance ($\rho = +0.81$), at a bounded cost to Frontier Distance
($\rho = -0.54$). The $t_{\text{start}}$ row exhibits strong negative
correlations with Nash alignment ($\rho = -0.88$), confirming that delaying
activation substantially reduces the corrective steps available before the
chain terminates. Together, these results establish a clear operational
directive: guidance requires early activation and large step sizes to maximize
normative compliance.

Three secondary effects merit explicit note. First, $\beta$ yields a
counterintuitive negative correlation with IR Compliance 
($\rho \approx -0.34$): excessively large IR penalty magnitudes produce
gradient norms that drive $\hat{\mathbf{u}}_0$ past the IR boundary into
Nash-suboptimal territory. Second, $\alpha$ is near-zero across all metrics
($|\rho| \leq 0.13$), confirming saturation. Third, additional DDIM steps
are negatively correlated with Nash alignment ($\rho = -0.31$), because more
unguided early steps compound latent drift before $\nabla_{\hat{\mathbf{u}}_0}
\mathcal{L}_{\text{guide}}$ can intervene.

\begin{figure}[h!]
    \centering
    \includegraphics[width=\textwidth]{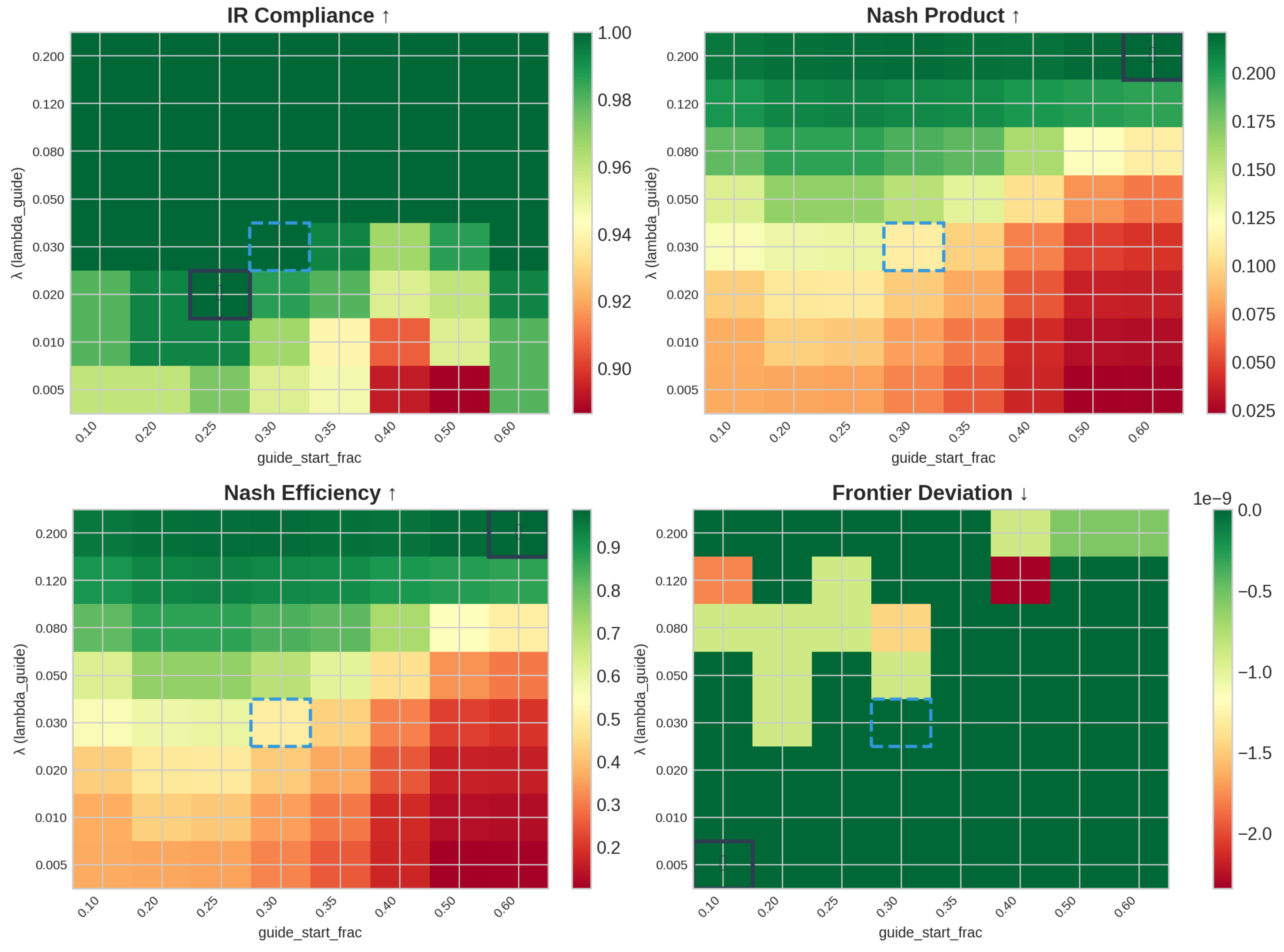}
    \caption{%
        Joint 2-D grid search over $\lambda$ (\texttt{lambda\_guide},
        y-axis) and $t_{\text{start}}$ (\texttt{guide\_start\_frac},
        x-axis) across four evaluation metrics (Synthetic NTU dataset).
        The dashed blue box marks the base configuration
        ($\lambda=0.03$, $t_{\text{start}}=0.30$); the solid box marks
        the composite-optimal configuration. Large $\lambda$ combined
        with early activation maximizes Nash Product and IR Compliance
        but increases Frontier Distance; small $\lambda$ with late
        activation preserves feasibility at the cost of Nash alignment,
        yielding efficiencies as low as $2.5\%$. No single cell
        dominates all four metrics simultaneously, confirming that a
        composite objective is required.%
    }
    \label{fig:sensitivity_heatmap}
\end{figure}

To investigate the interaction between the two dominant parameters, a joint
$8{\times}8$ grid search over $(\lambda,\, t_{\text{start}})$ was conducted
(Figure~\ref{fig:sensitivity_heatmap}). The IR Compliance panel reveals a
broad high-performance region (${\geq}0.98$) in the upper-left quadrant,
degrading monotonically toward the lower-right. Nash Product and Efficiency
show a sharper boundary: only a narrow corridor in the upper-left achieves
high performance, with Nash Efficiency ranging from $2.5\%$ (small $\lambda$,
late $t_{\text{start}}$) to $20\%$ (large $\lambda$, early $t_{\text{start}}$).
The base configuration (dashed box) sits at ${\approx}10\%$ efficiency —
substantially below the achievable optimum. The Frontier Distance panel
exhibits the inverse spatial pattern, exposing the fundamental trade-off: no
single cell simultaneously dominates all four metrics, confirming the necessity
of a composite objective.

The optimal configurations derived from this composite objective are reported
in Table~\ref{tab:optimal_params}. Because $\lambda$ exhibits a perfect
monotone relationship with Nash optimization, the composite objective selects the
upper bound of the search grid ($\lambda=0.35$) across all three datasets, indicating that
$\lambda$ is a domain-invariant parameter. We note that performance did not
saturate at this boundary, suggesting the true optimum may exceed $0.35$.
Consistently fewer DDIM steps ($S = 10$--$15$, versus the default $S = 50$)
are selected, mitigating early-chain noise accumulation before guidance
activates, consistent with $\rho(S,\,\text{Nash}) = -0.31$.

\begin{table}[htbp]
\centering
\caption{Optimal Hyperparameter Configurations Derived via Sensitivity Analysis}
\label{tab:optimal_params}
\resizebox{\textwidth}{!}{%
\renewcommand{\arraystretch}{1.3}
\begin{tabular}{lcccccc}
\toprule
\textbf{Dataset}
  & $\boldsymbol{\lambda}$ \small{(\texttt{lambda\_guide})}
  & $\boldsymbol{t_{\text{start}}}$ \small{(\texttt{guide\_start\_frac})}
  & $\boldsymbol{S}$ \small{(\texttt{steps})}
  & $\boldsymbol{\alpha}$ \small{(\texttt{alpha\_norm})}
  & $\boldsymbol{\beta}$ \small{(\texttt{beta\_ir})}
  & $\boldsymbol{\gamma}$ \small{(\texttt{gamma\_frontier})} \\
\midrule
Synthetic NTU    & $0.35$ & $0.25$ & $15$ & $10.0$  & $8.0$  & $15.0$ \\
CaSiNo Corpus    & $0.35$ & $0.40$ & $10$ & $150.0$ & $20.0$ & $10.0$ \\
Deal or No Deal  & $0.35$ & $0.30$ & $10$ & $10.0$  & $0.5$  & $10.0$ \\
\bottomrule
\multicolumn{7}{l}{\footnotesize
  All symbols defined in Table~\ref{tab:notation};
  search domains in Table~\ref{tab:search_space}.
  Values selected by composite objective:
  $0.40 \cdot \text{IR} + 0.35 \cdot \text{Efficiency}
   + 0.15 \cdot \text{Nash} + 0.10 \cdot (1 - \text{FDist})$.}
\end{tabular}%
}
\end{table}

Domain-specific structural differences govern the remaining parameters. The selection of
$t_{\text{start}}$ ($0.25$ for Synthetic NTU versus $0.40$ for CaSiNo) is
structurally significant: empirical human corpora exhibit higher variance and produce less stable
early estimates $\hat{\mathbf{u}}_0$, requiring delayed activation to avoid
corrupting the latent trajectory before the signal is reliable. The large
values of $\alpha$ ($150.0$) and $\beta$ ($20.0$) selected for CaSiNo reflect
a fundamental property of empirical negotiation data: human negotiators
frequently accept interior-point outcomes far from the Pareto efficient, biasing
the learned distribution toward suboptimal utility regions and demanding
stronger normative correction. Conversely, the low $\beta = 0.5$ for
Deal-or-No-Deal is consistent with that dataset's discrete item-allocation
structure, where utility vectors are naturally bounded away from the
disagreement boundary.

\subsection{Axiomatic Compliance and Efficiency Gains}
\label{subsec:compliance}

Using the optimized configurations from Table~\ref{tab:optimal_params}, the
guided framework is benchmarked against an unguided DDIM baseline — identical
architecture, identical training, guidance disabled at inference. Results are
reported in Table~\ref{tab:performance} and spatial distributions are
visualized in Figure~\ref{fig:utility_allocations}.

\begin{table}[htbp]
\centering
\caption{Quantitative Performance: Unguided vs.\ Guided DDIM}
\label{tab:performance}
\renewcommand{\arraystretch}{1.3}
\resizebox{\textwidth}{!}{%
\begin{tabular}{llcccc}
\toprule
\textbf{Dataset}
  & \textbf{Model}
  & \textbf{IR Compliance $\uparrow$}
  & \textbf{Nash Product $\uparrow$}
  & \textbf{Nash Efficiency $\uparrow$}
  & \textbf{Frontier Dist.\ $\downarrow$} \\
\midrule
\multirow{2}{*}{Synthetic NTU}
  & Unguided        & $0.9390$ & $0.0567$ & $24.29\%$ & $0.0000$ \\
  & \textbf{Guided} & $\mathbf{1.0000}$ & $\mathbf{0.2323}$
                    & $\mathbf{99.45\%}$ & $\mathbf{0.0000}$ \\
\midrule
\multirow{2}{*}{CaSiNo Corpus}
  & Unguided        & $1.0000$ & $0.0202$ & $05.64\%$ & $0.0000$ \\
  & \textbf{Guided} & $\mathbf{1.0000}$ & $\mathbf{0.2712}$
                    & $\mathbf{54.24\%}$ & $\mathbf{0.0000}$ \\
\midrule
\multirow{2}{*}{Deal or No Deal}
  & Unguided        & $1.0000$ & $0.3374$ & $67.47\%$ & $0.0153$ \\
  & \textbf{Guided} & $\mathbf{1.0000}$ & $\mathbf{0.4434}$
                    & $\mathbf{88.67\%}$ & $\mathbf{0.0000}$ \\
\bottomrule
\multicolumn{6}{l}{\footnotesize
  Statistical significance (Wilcoxon signed-rank, Nash Product):
  Synthetic NTU $p = 3.33 \times 10^{-165}$;\;
  CaSiNo $p = 6.91 \times 10^{-149}$;\;
  Deal or No Deal $p = 9.24 \times 10^{-42}$.}
\end{tabular}%
}
\end{table}

\begin{figure}[htbp]
    \centering
    \includegraphics[width=\textwidth]{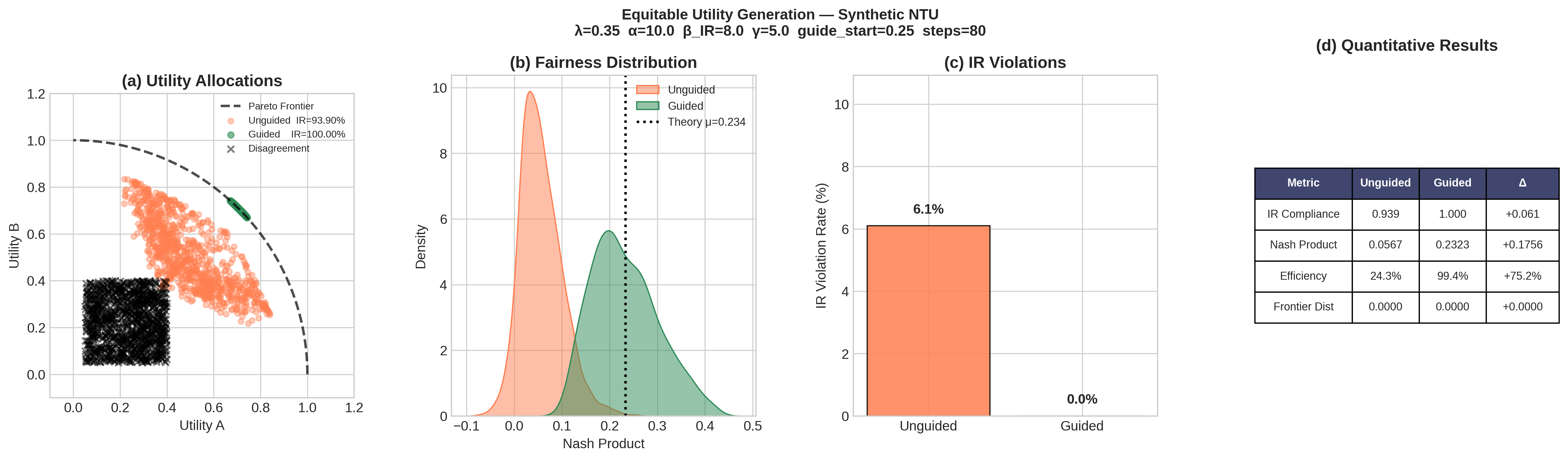} \\
    \vspace{0.4cm}
    \includegraphics[width=\textwidth]{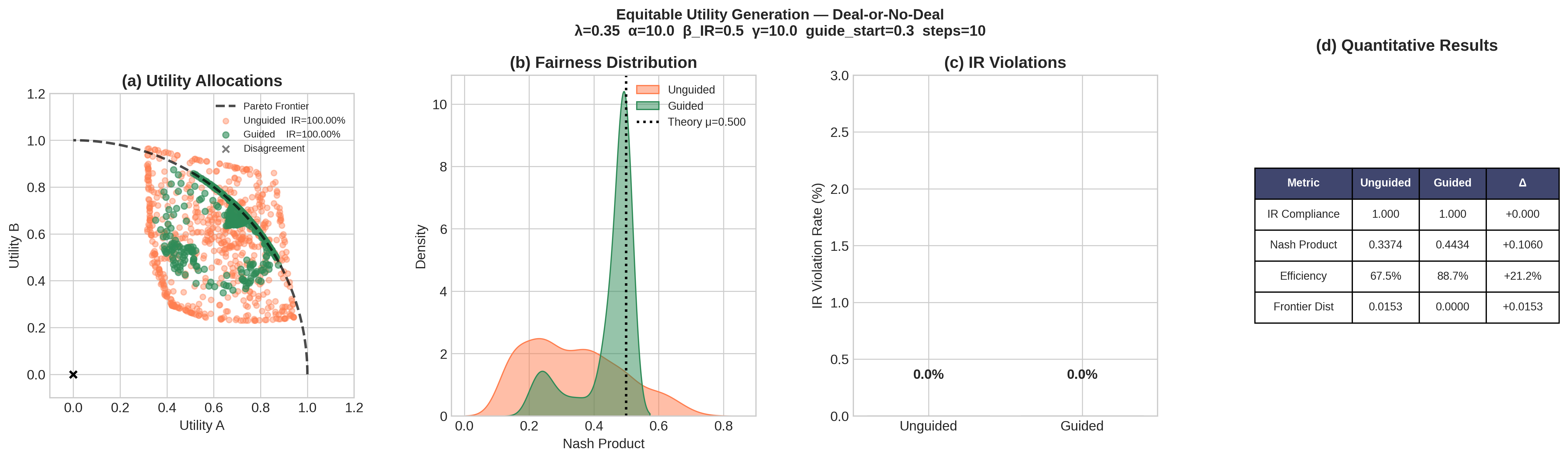} \\
    \vspace{0.4cm}
    \includegraphics[width=\textwidth]{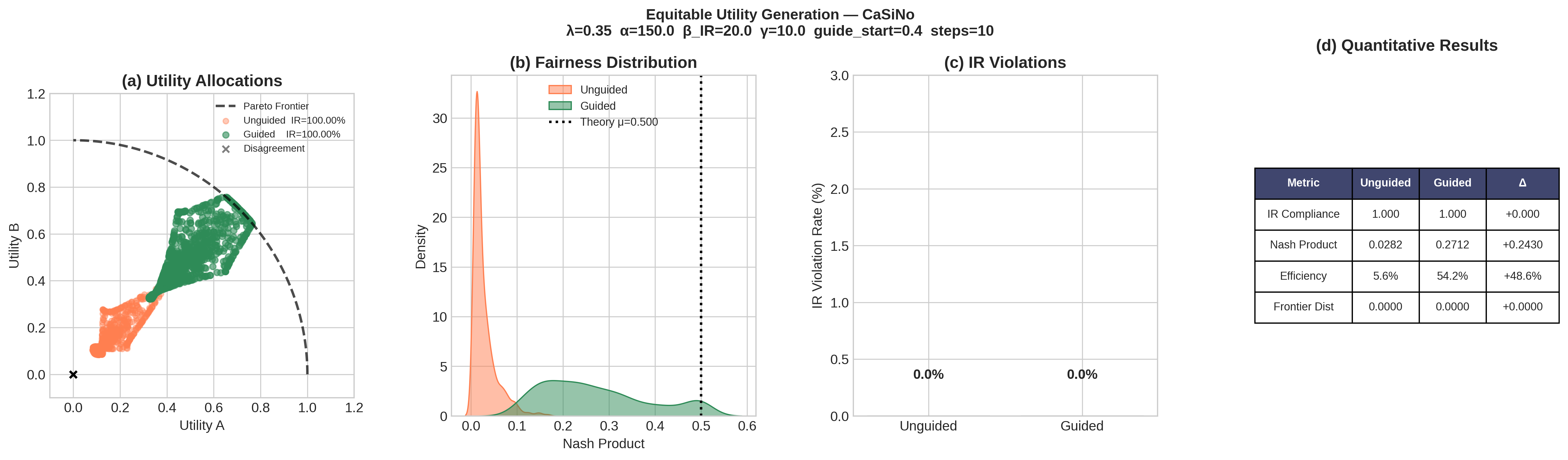}
    \caption{%
        Visual evaluation of generated utility allocations across the three
        negotiation domains under optimized guidance configurations
        (Table~\ref{tab:optimal_params}).
        \textbf{Panel (a):} Spatial distribution of generated utilities in
        $[0,1]^2$; the guided framework (green) actively repels from the
        disagreement points $\mathbf{d}$ (black crosses) while concentrating
        mass near the Pareto frontier arc.
        \textbf{Panel (b):} Nash product density; rightward shift indicates
        improved joint surplus and fairness.
        \textbf{Panels (c) and (d):} IR Compliance and Nash Efficiency
        distributions quantifying the reduction in axiomatic violations and
        recovery of operational optimality relative to the unguided baseline.%
    }
    \label{fig:utility_allocations}
\end{figure}

The spatial distributions in Figure~\ref{fig:utility_allocations}(a) expose
the limitations of the unguided baseline. In the Synthetic NTU domain, the
unguided model exhibits higher variance, with samples crossing the IR
boundary (black crosses), corresponding to the $6.1\%$ IR violation rate in
Table~\ref{tab:performance}. In the human-centric datasets, the unguided
model exhibits mode collapse toward the origin: because empirical training data
reflects cautious, suboptimal human bargaining, the learned distribution
concentrates near $(0, 0)$. Although this conservative clustering avoids IR
violations, it fails to optimize the negotiation — yielding $5.6\%$ Nash
Efficiency on CaSiNo.

Integration of the differentiable guidance loss $\mathcal{L}_{\text{guide}}$
transforms the generative distribution. The guided
allocations (green) repel from $\mathbf{d}$, ensuring $1.000$ IR
Compliance across all datasets. The Nash gradient ($\alpha$ term) stretches
the distribution toward the Pareto arc, escaping the biased empirical prior.
This geometric expansion is captured quantitatively in
Table~\ref{tab:performance}: Synthetic NTU efficiency rises from $24.3\%$ to
$99.4\%$; CaSiNo improves nearly tenfold from $5.6\%$ to $54.2\%$; and
Deal-or-No-Deal advances from $67.5\%$ to $88.7\%$, simultaneously
eliminating the residual frontier distance of $0.0153$ left by the unguided
model.

The significance of these improvements is confirmed by the Wilcoxon
signed-rank test (reported in the footnote of Table~\ref{tab:performance}).
Across all three datasets, $p$-values are highly significant ($p < 0.001$),
formally establishing that the guided and unguided Nash product distributions
are statistically distinct and that the improvement is not attributable to
random variation.

\subsection{Latent Trajectory Dynamics and Stabilization}
\label{subsec:dynamics}

Beyond final-outcome metrics, step-wise analysis of the latent state
$\mathbf{u}_t$ provides mechanistic insight into how the guidance framework
achieves normative compliance. Aggregate trajectory statistics (mean $\pm 1$
standard deviation across 30 test cases) are shown in
Figure~\ref{fig:traj_aggregate}; individual case trajectories are examined in
Figure~\ref{fig:dynamics_synthetic}.

\begin{figure}[h!]
    \centering
    \includegraphics[width=0.95\textwidth]{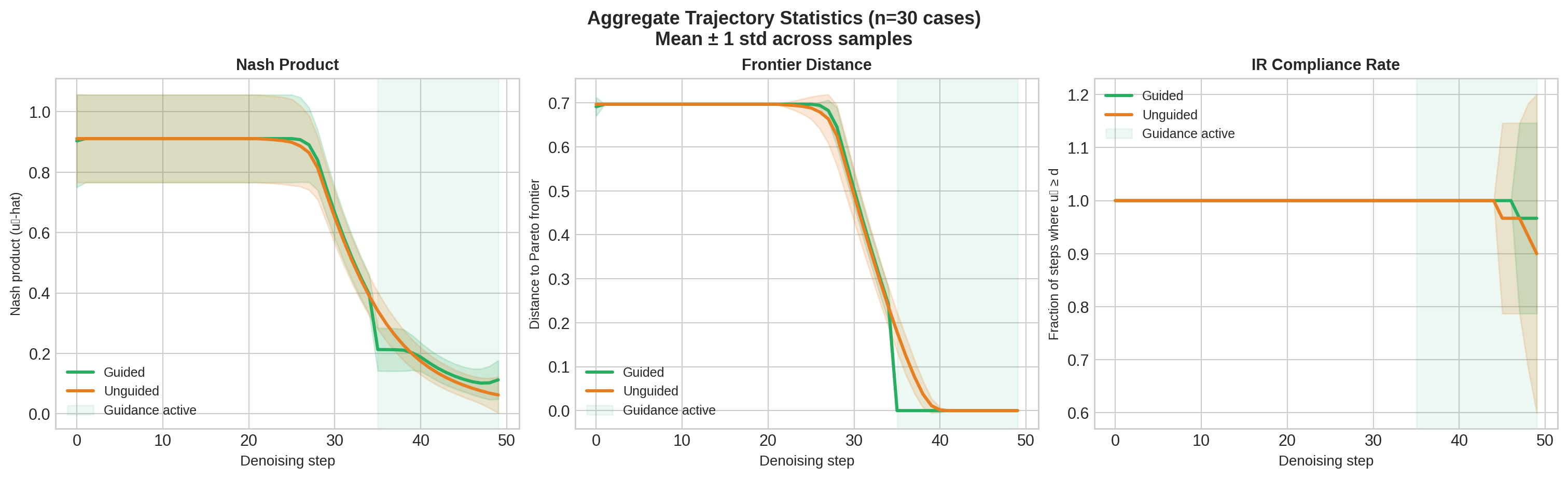}\\
    \vspace{0.3cm}
    \includegraphics[width=\textwidth]{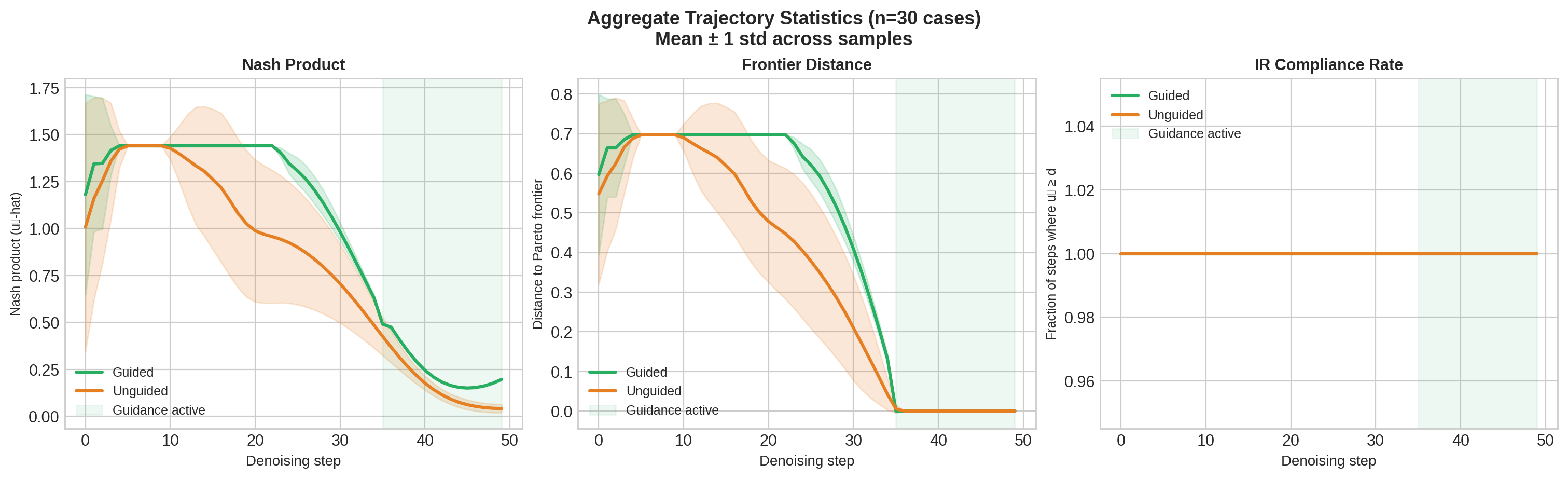}\\
    \vspace{0.3cm}
    \includegraphics[width=\textwidth]{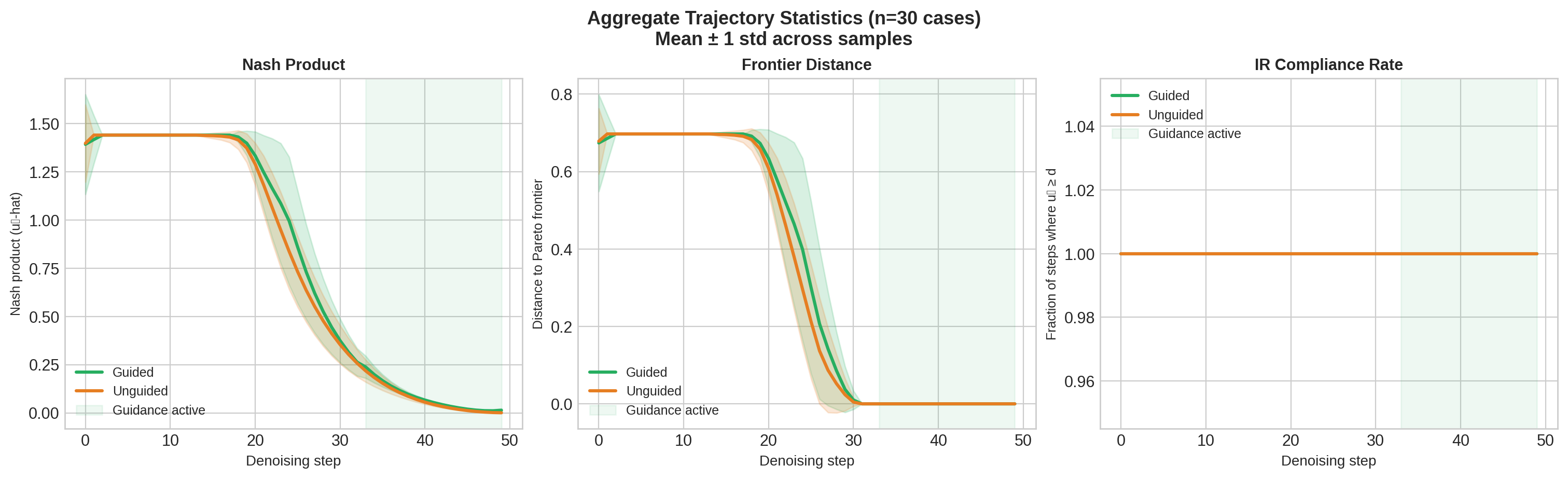}
    \caption{%
        Aggregate trajectory statistics (mean $\pm 1$ std, $n=30$ test cases)
        for guided (green) and unguided (orange) DDIM chains across the three
        negotiation domains. The green shaded band indicates the active
        guidance window ($t/T < t_{\text{start}}$).
        \textbf{Top (Synthetic NTU):} IR Compliance (right panel) shows the
        unguided model suffering from late-stage cumulative drift, while the
        guided chain maintains strict $1.000$ compliance throughout the window.
        \textbf{Middle (Deal or No Deal):} Nash Product (left panel) diverges
        from the unguided baseline upon entering the guidance window, arresting
        joint surplus decay.
        \textbf{Bottom (CaSiNo):} Frontier Distance (center panel) shows
        synchronized convergence to zero for both modes, confirming that the
        MLP denoiser independently places samples near the Pareto arc;
        guidance provides \emph{directional} (Nash-optimal) correction, not
        general feasibility correction.%
    }
    \label{fig:traj_aggregate}
\end{figure}

Without the two stabilization mechanisms described in
Section~\ref{subsec:diffusion}, applying the log-barrier gradient
$\nabla_{\hat{\mathbf{u}}_0} \mathcal{L}_{\text{guide}}$ directly to
the DDIM state induces severe trajectory divergence. As
$\hat{\mathbf{u}}_0$ approaches the IR boundary
($u_i \to d_i$), the Nash term $\mathcal{L}_N$ produces unbounded gradients.
Empirical tracking confirmed that without the soft-clamp
(Fix-B: $\mathbf{u}_t \leftarrow \operatorname{clip}(\mathbf{u}_t,
-c_{\text{drift}}, c_{\text{drift}})$), the latent state accumulates up to
$3.0$ units of negative drift over 22 guidance steps, collapsing IR Compliance
to $0\%$ and generating out-of-distribution utility estimates. Additionally,
returning the DDIM output rather than $\hat{\mathbf{u}}_0^{\text{guided}}$ at
the terminal step (the pre-Fix-A behavior) discards all normative corrections,
since $\mathbf{u}_{\text{out}} = \hat{\boldsymbol\varepsilon}$ when
$\bar\alpha_{\text{next}} = 0$.

With both mechanisms applied, gradient-normalized guidance maintains structural integrity throughout the
reverse process. As observed in the aggregate Nash Product panels
(Figure~\ref{fig:traj_aggregate}, left column), the guided and unguided
$\hat{\mathbf{u}}_0$ estimates follow identical trajectories during the
initial unguided steps — confirming that the base denoiser is undisturbed.
Upon crossing $t/T = t_{\text{start}}$, the gradient norm
$\|\nabla_{\hat{\mathbf{u}}_0} \mathcal{L}_{\text{guide}}\|_2$ activates and
grows monotonically. The latent state smoothly redirects from the biased
empirical mean, halting joint surplus decay and converging to a strictly
higher Nash product. The symmetric soft-clamp mitigates residual numerical
overflow, preventing the compounding drift that would otherwise violate utility
boundary constraints.

\begin{figure}[h!]
    \centering
    \includegraphics[width=\textwidth]{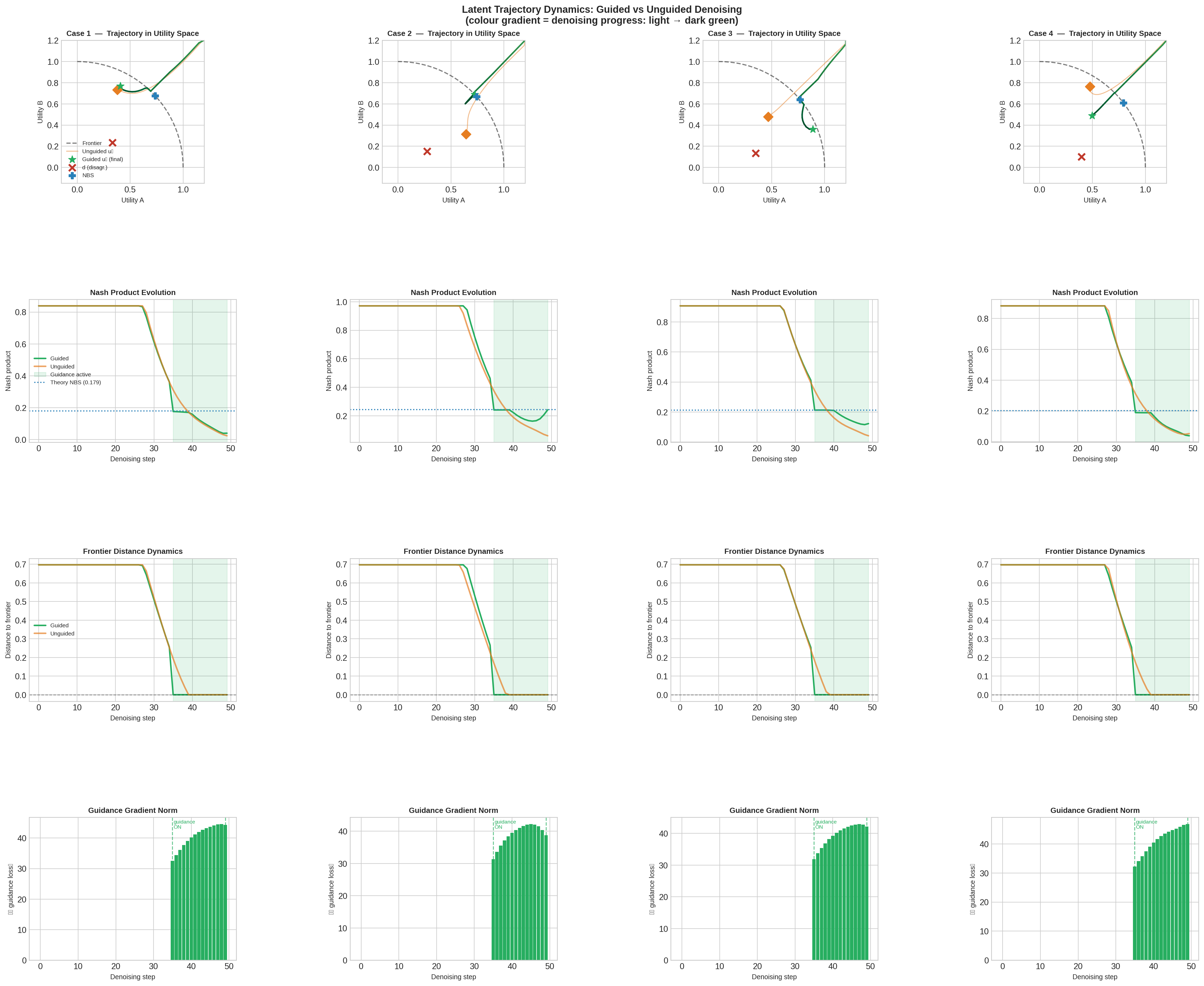}
    \caption{%
        \textbf{Latent trajectory dynamics — 4 representative cases
        (Synthetic NTU, $S=50$ visualization steps).}
        Each column is an independent negotiation scenario; guided (green)
        and unguided (orange) trajectories share identical noise
        initializations so that differences are attributable solely to
        $\nabla_{\hat{\mathbf{u}}_0} \mathcal{L}_{\text{guide}}$.
        \textbf{Row 1:} Utility-space path of $\hat{\mathbf{u}}_0$ over
        denoising steps; color intensity encodes progress (light
        $\to$ dark). Stars ($\star$) and diamonds ($\diamond$) mark guided
        and unguided terminal allocations; the dashed arc is the Pareto
        frontier; red crosses mark disagreement points $\mathbf{d}$.
        \textbf{Row 2:} Nash product $\prod_i(\hat{u}_{0,i} - d_i)$ per
        step; green band = guidance window; dotted line = theoretical NBS.
        \textbf{Row 3:} Distance of $\hat{\mathbf{u}}_0$ to the Pareto
        frontier — both modes converge to ${\approx}0$, confirming that
        the MLP denoiser achieves frontier proximity independently;
        guidance provides Nash-directional, not feasibility, correction.
        \textbf{Row 4:} Guidance gradient norm
        $\|\nabla_{\hat{\mathbf{u}}_0} \mathcal{L}_{\text{guide}}\|_2$,
        zero outside the guidance window and growing monotonically within
        it due to the adaptive weighting $w = \lambda\sqrt{\bar{\alpha}_t}$.%
    }
    \label{fig:dynamics_synthetic}
\end{figure}

Individual trajectory analysis (Figure~\ref{fig:dynamics_synthetic}) provides
a micro-level view of the geometric steering mechanism across four test cases.
Tracking $\hat{\mathbf{u}}_0$ in utility space (Row~1) reveals how guidance
corrects off-frontier drift: in Case~2, the unguided trajectory terminates
well inside the feasible set at a suboptimal Nash product, while the guided
trajectory steers onto the Pareto arc. In Case~4, where the unguided model
overshoots the feasible set, the projection $\Pi_{\mathcal{F}}$
restores feasibility.

The Nash product evolution (Row~2) highlights temporal dynamics: upon
entering the guidance window (green band), the guided trajectory consistently
terminates closer to the theoretical NBS (dotted line) than the unguided
baseline. Row~3 confirms that Frontier Distance converges to near zero for
both modes — establishing that the MLP denoiser $s_\theta$ independently
achieves Pareto proximity; the role of $\mathcal{L}_{\text{guide}}$ is
directional (Nash-optimal position on the frontier) rather than general
feasibility enforcement. The gradient norm (Row~4) grows monotonically from
step 35 to 50, consistent with the adaptive weighting $w = \lambda
\sqrt{\bar{\alpha}_t}$: as $\bar{\alpha}_t$ increases toward unity in the
low-noise regime, guidance strength increases precisely when the clean
estimate $\hat{\mathbf{u}}_0$ is most reliable.

\subsection{Comparison with Baseline Methods}
\label{subsec:sota}

We evaluate the guided graph diffusion framework against six baseline approaches, encompassing direct regression models, conventional generative architectures, and structural ablations of the proposed system. All baselines are trained and assessed on identical dataset partitions (80/10/10 train/validation/test, random seed 42), with utilities normalized to $[0,1]^2$, and performance metrics computed over 500 held-out test instances. Table~\ref{tab:sota_comparison} reports IR Compliance and Nash Efficiency for each method across the three negotiation domains.

\begin{table}[htbp]
\centering
\caption{%
  Comparison of baseline methods on IR Compliance (\%) and Nash Efficiency (\%).
  Higher is better for both metrics.
  Nash Efficiency is normalized per-instance by the SLSQP solution
  $\mathbf{u}^*_{\text{NBS}}$; a score of 100\% indicates exact recovery of NBS.
  $^\dagger$~SLSQP requires full frontier knowledge at inference time and is
  not deployable; its Nash Efficiency is 100\% by definition (serving as the
  normalization reference) and is shown for completeness only.
}
\label{tab:sota_comparison}
\renewcommand{\arraystretch}{1.3}
\resizebox{\textwidth}{!}{%
\begin{tabular}{l|cc|cc|cc}
\toprule
\multirow{2}{*}{\textbf{Method}}
  & \multicolumn{2}{c|}{\textbf{Synthetic NTU}}
  & \multicolumn{2}{c|}{\textbf{CaSiNo}}
  & \multicolumn{2}{c}{\textbf{Deal or No Deal}} \\
  & \textbf{IR (\%)} & \textbf{Nash Eff.\ (\%)}
  & \textbf{IR (\%)} & \textbf{Nash Eff.\ (\%)}
  & \textbf{IR (\%)} & \textbf{Nash Eff.\ (\%)} \\
\midrule
SLSQP Oracle$^\dagger$~\cite{2020scipy}
  & 100.0 & \emph{def.}
  & 100.0 & \emph{def.}
  & 100.0 & \emph{def.} \\
\midrule
Supervised MLP
  & 100.0 &  40.10
  & 100.0 &   0.03
  & 100.0 &  61.83 \\
CVAE \cite{sohn2015learning}
  &  98.4 &  33.88
  & 100.0 &   0.03
  & 100.0 &  62.46 \\
Conditional GAN \cite{mirza2014conditional}
  & 100.0 &  39.63
  & 100.0 &   0.00
  & 100.0 &  68.41 \\
\midrule
Unguided DDIM (ours, no guidance)
  &  93.9 &  24.29
  & 100.0 &   5.64
  & 100.0 &  67.47 \\
Projection DDIM (ours, post-hoc)
  & 100.0 &  50.55
  & 100.0 &   5.90
  & 100.0 &  65.72 \\
Hard-constraint DDIM (ours, $t_{\text{start}}=1$)
  &  99.8 &  59.29
  & 100.0 &  45.66
  & 100.0 &  88.12 \\
\midrule
\textbf{Guided DDIM (ours, $t_{\text{start}} < 1$)}
  & \textbf{100.0} & \textbf{99.45}
  & \textbf{100.0} & \textbf{54.24}
  & \textbf{100.0} & \textbf{88.67} \\
\bottomrule
\end{tabular}%
}
\end{table}

As established in Section~\ref{subsec:compliance}, the guided graph diffusion model achieves 100\% IR compliance and the highest Nash Efficiency across all domains. Table~\ref{tab:sota_comparison} contextualizes this performance against baseline methods. On the Synthetic NTU dataset, the guided model ($99.45\%$) outperforms the strongest deep generative baseline (Conditional GAN, $39.63\%$) by 59.8\% points. On the empirical corpora (CaSiNo and Deal or No Deal), the guided model maintains this performance advantage while preserving strict IR compliance, demonstrating an adherence to axiomatic boundaries not observed in the unguided generative baselines.

To understand the source of this performance gap, we first examine the behavior of established generative architectures operating without normative constraints. On the Synthetic NTU dataset, the Conditional GAN achieves $39.63\%$ Nash Efficiency and the CVAE $33.88\%$. However, on the CaSiNo corpus, both architectures yield near-zero efficiency (CVAE $0.03\%$, Conditional GAN $0.00\%$). This reduction reflects a structural limitation of unconstrained generative models: they simply replicate the suboptimal bargaining distributions present in empirical training data, producing allocations that are individually rational but Nash-inefficient. The guided model mitigates this reliance on empirical priors via its differentiable normative loss, achieving $54.24\%$ on the same corpus. On Deal or No Deal, the GAN achieves $68.41\%$ and the CVAE $62.46\%$, remaining $20.3$ and $26.2$\% points below the guided model ($88.67\%$). Furthermore, neither of the generative baselines ensures strict IR compliance across all tested domains.

Beyond probabilistic generation, direct regression approaches face similar empirical pitfalls. The Supervised MLP baseline constitutes a direct regression model that utilizes the same GATv2 encoder, conditioned on agent-specific features. It attains $40.10\%$ efficiency on Synthetic NTU and $61.83\%$ on Deal or No Deal, but drops to only $0.03\%$ on CaSiNo. Owing to its training via a mean squared error (MSE) loss with respect to human-generated utility targets, the MLP predominantly learns to reproduce average empirical human allocation behaviors, including systematic under-claiming. These findings suggest that the observed performance gap between supervised methods and the guided model on human datasets stems from the fundamental distinction between imitating realized human outcomes and explicitly optimizing with respect to a normative objective.

Having established the limitations of external baselines, three structural variants of the proposed method are examined to isolate the contributions of key internal design choices. First, the Unguided DDIM baseline ($24.29\%$, $5.64\%$, $67.47\%$) quantifies the effect of the denoiser in isolation: it generates samples that are close to the Pareto frontier but does not implement any Nash-directional adjustment. Second, the Projection DDIM baseline introduces guidance only as a single post-processing operation, in which the final unguided sample is projected onto the Pareto frontier using SLSQP. On CaSiNo, this configuration attains $5.90\%$ efficiency, which is comparable to the unguided baseline ($5.64\%$). This outcome indicates that iterative guidance is essential; a one-time post-hoc correction applied to an empirically biased sample is insufficient to achieve the optimal projection target. Third, the Hard-constraint DDIM variant applies guidance at every reverse diffusion step ($t_{\text{start}} = 1$). On Synthetic NTU, it attains $59.29\%$ efficiency, which is substantially lower than the full model’s $99.45\%$, and it induces a reduction in IR compliance to $99.8\%$. This degradation empirically supports the use of the $t_{\text{start}}$ activation window: an overly restrictive gradient field at high noise levels interferes with mathematical feasibility. Constraining guidance to only the final fraction of diffusion steps permits the denoiser to first attain Pareto proximity before the Nash-directional correction is enforced. On Deal or No Deal, the efficiency difference is markedly smaller ($88.12\%$ vs.\ $88.67\%$), suggesting that the dataset’s discrete item structure reduces its sensitivity to the specific guidance activation schedule.

Finally, the normative superiority of the guided framework must be weighed against its computational demands. Generative baseline models (CVAE, Conditional GAN, and Supervised MLP) require only a single forward pass through the neural network at inference time. In contrast, the guided DDIM procedure necessitates $S \in [10, 15]$ sequential denoising steps, leading to a commensurate increase in computational latency. Nonetheless, on empirical datasets such as CaSiNo, where single-step generative approaches exhibit near-zero normative efficiency, this additional computational burden constitutes a necessary trade-off of inference speed for enhanced normative robustness. In particular, attaining a Nash Efficiency of $54.24\%$ while simultaneously guaranteeing $100\%$ IR compliance was not achieved by any single-pass architecture evaluated within this framework.

\section{Discussion}
\label{sec:discussion}

The empirical results validate the core hypothesis that equipping deep generative models with differentiable normative guidance effectively overcomes historical biases in human bargaining data. By decoupling relational state encoding from generative search, the guided graph diffusion framework bridges the gap between the strict axiomatic guarantees of cooperative game theory and the representational capacity of modern deep learning. This architecture directly addresses the structural deficiencies that have historically limited autonomous negotiation systems.

The most consequential finding is the severe efficiency degradation of standard generative baselines on the human centric CaSiNo corpus. While these unconstrained models perform adequately on synthetic datasets with exact Pareto frontiers, they collapse to near zero efficiency when trained on actual human negotiations. This dichotomy highlights a structural limitation of purely data-driven agents. Human bargaining is boundedly rational \cite{simon1955behavioral}, characterized by risk aversion \cite{kahneman1979prospect}, anchoring \cite{tversky1974judgment}, and a tendency to accept suboptimal interior point outcomes. Consequently, training corpora densely record biased allocation behavior.

When standard architectures minimize divergence objectives against these datasets, they merely replicate this suboptimality. The supervised baseline exhibits the same failure mode through mean squared error regression. Our framework circumvents this trap by utilizing empirical data solely to map the feasible set while relying on the guidance loss to navigate toward the Pareto efficient. The massive efficiency gap between the guided model and unconstrained baselines confirms that differentiable steering is essential for recovering high surplus outcomes when human behavior cannot serve as a normative target.

Ablation studies illuminate the mechanical advantage of iterative guided diffusion over naive constraint enforcement. The poor performance of the post hoc projection baseline demonstrates that distributional collapse is structurally embedded in the generative sample. A single radial projection onto the Pareto frontier cannot recover optimal joint surplus because it lacks the Nash directional information required to locate the correct point on the frontier arc.

Guided diffusion addresses this through a two phase dynamic. Restricting guidance to the final low noise steps allows the unconditional score network to first establish a Pareto proximate distribution. Subsequently, normative gradients execute fine grained directional optimization. This directional correction requires only the gradient of the guidance loss computed from the disagreement vector, avoiding the need for explicit analytical knowledge of the Pareto frontier geometry.

Deploying autonomous agents in high-stakes domains like supply chain procurement \cite{carbonneau2008application} or legal dispute resolution \cite{lodder2005negotiation} requires strict mathematical guarantees. A system must ensure that no generated outcome is worse than the best alternative to a negotiated agreement for any agent \cite{fisher1981getting}. Our framework demonstrates exceptional reliability in maintaining individual rationality compliance across the evaluated domains, without requiring hard constraint projections that over constrain the trajectory. Instead, the softplus based penalty provides a continuous repulsive force away from the boundary. Furthermore, because the encoder explicitly models relational asymmetries, the framework naturally supports counterfactual fairness audits \cite{kusner2017counterfactual}, allowing practitioners to quantify how structural imbalances propagate to final allocations.

Despite its strong empirical performance, the framework has boundary conditions that inform future research. First, the sequential iterative denoising adds minor inference latency, precluding deployment in ultra-low latency environments like high-frequency automated bidding. Second, the formulation assumes fully observable disagreement vectors and priority weights. Extending the guidance mechanism to operate under uncertainty via probabilistic belief distributions represents a critical next step \cite{barry2000lying}.

Third, the optimal guidance step size hit the upper boundary of the search grid, suggesting that expanding the parameter sweep may yield even higher efficiencies. Fourth, while the graph encoder naturally scales to multilateral settings involving multiple parties, the geometric dynamics of guidance gradients in higher dimensional utility spaces require rigorous empirical validation. Finally, the frontier distance metric used in evaluation only penalizes outward overshoot beyond the unit ball, failing to discriminate between optimal frontier samples and those concentrated near the origin. Consequently, Nash Efficiency remains the sole reliable indicator of geometric progress toward the theoretical optimum.

\section{Conclusion and Future Work}
\label{sec:conclusion}

This paper introduced a guided graph diffusion framework for constructing equitable utility vectors in bilateral automated negotiation. The framework combines a graph attention encoder that captures asymmetric agent attributes in a relational representation, a conditional diffusion model that generates utility vectors in continuous space, and a composite normative guidance loss that enforces individual rationality and maximizes joint surplus throughout the reverse process. A design utilizing late activation withholds normative corrections until the final stage of generation. This allows the model to establish a structure proximate to the Pareto frontier before applying directional refinements toward the NBS.

Evaluation across three datasets confirms the effectiveness of this approach. On synthetic data with exactly computable optima, the framework achieves 99.45\% Nash Efficiency and 100\% individual rationality compliance. This performance falls within 0.55\% points of an oracle solver that requires full Pareto frontier knowledge. On empirical corpora derived from human negotiators, unconstrained generative baselines collapse to near zero efficiency. In contrast, the proposed framework achieves 54.24\% and 88.67\% Nash Efficiency on the CaSiNo and Deal or No Deal datasets respectively, while maintaining perfect individual rationality compliance across all evaluation domains. Ablation studies confirm that the activation timing of the guidance window is the most influential architectural parameter. Applying normative corrections from the start of generation degrades Nash Efficiency by up to 40\% points and introduces individual rationality violations.

Two limitations bound the scope of the current framework. First, the complete information assumption requires the disagreement point and priority weights of each agent to be known at inference time, which may not hold in adversarial settings. Second, the sequential denoising process introduces inference latency relative to single-pass baselines. This tradeoff is justified by the normative gains achieved on biased corpora but remains a practical consideration for real time deployment.

Future work will pursue two primary extensions. The framework will be extended from bilateral to multilateral negotiations. While the graph attention encoder generalizes to graphs of arbitrary size, the denoiser output dimension and the guidance loss projection are currently designed for two agents. Extending the architecture to multiple parties requires redesigning both components alongside managing the increased gradient variance associated with a higher dimensional Nash product. Furthermore, the guidance mechanism will be adapted for incomplete information settings by operating over a probabilistic belief state for the hidden attributes of each agent. This adaptation will enable deployment in environments where outside options and priorities are only partially observable.

\section*{Declaration of Competing Interest}
The author declares that they have no known competing financial interests or personal relationships that could have appeared to influence the work reported in this paper.

\section*{Declaration of Generative AI and AI Assisted Technologies in the Writing Process}
During the preparation of this work the author used generative artificial intelligence technologies solely to improve language clarity and readability. After using these tools, the author reviewed and edited the content as needed and takes full responsibility for the final content of the publication.

\section*{Data Availability}
The empirical datasets analyzed during the current study, specifically the CaSiNo human negotiation corpus and the Deal or No Deal corpus, are publicly available from their respective original creators. The synthetic utility dataset generated during the current study is available from the corresponding author on reasonable request.

\section*{Code Availability}
The computational code required to reproduce the guided graph diffusion framework, the baseline models, and the experimental results presented in this study is available from the corresponding author on reasonable request.

\bibliographystyle{elsarticle-num}
\bibliography{ref}
\end{document}